%% file: Aspects_of_Asynchronicity.tex
\documentclass[a4paper,UKenglish,thm-restate]{lipics-v2021}
\usepackage{ltl}

\usepackage{xspace}
\usepackage{xcolor,stmaryrd}
\usepackage{scalerel,multicol}
\usepackage{mathtools}            
\usepackage{tabularx}
\usepackage{multirow}
\usepackage{makecell}
\usepackage{thmtools, comment, url} 
\usepackage{ifthen}
\usepackage{booktabs, todonotes}
\usepackage{hyperref}
\usepackage{multicol}

\DeclareMathOperator\dep{\mathrm{dep}}

\newcommand{\FO}{\mathrm{FO}}

\newcommand{\Max}[1]{\todo[inline]{Max: #1}}
\newcommand{\Jonni}[1]{\todo[inline, color=red!20]{Jonni: #1}}
\newcommand{\Juha}[1]{\todo[inline, color=green!20]{Juha: #1}}
\nolinenumbers

\input{macros}

\pdfoutput=1 

\title{Set Semantics for Asynchronous TeamLTL:
Expressivity and Complexity}

\author{Juha Kontinen}{University of Helsinki, Finland}{juha.kontinen@helsinki.fi}{https://orcid.org/0000-0003-0115-5154}{Partially supported by the Academy of Finland grant 345634.}

\author{Max Sandstr\"om\footnote{Corresponding author}}{University of Sheffield, UK \and University of Helsinki, Finland}{m.sandstrom@sheffield.ac.uk}{https://orcid.org/0000-0002-6365-2562}{Partially supported by the Academy of Finland grant 322795.}

\author{Jonni Virtema}{University of Sheffield, UK \and University of Helsinki, Finland}{j.t.virtema@sheffield.ac.uk}{https://orcid.org/0000-0002-1582-3718}{Partially supported by the Academy of Finland grant 345634 and by the DFG grant VI 1045/1-1.}

\authorrunning{J. Kontinen, M. Sandstr\"om and J. Virtema}

\Copyright{Juha Kontinen, Max Sandstr\"om and Jonni Virtema}



\category{}

\EventEditors{John Q. Open and Joan R. Access}
\EventNoEds{2}
\EventLongTitle{42nd Conference on Very Important Topics (CVIT 2016)}
\EventShortTitle{CVIT 2016}
\EventAcronym{CVIT}
\EventYear{2016}
\EventDate{December 24--27, 2016}
\EventLocation{Little Whinging, United Kingdom}
\EventLogo{}
\SeriesVolume{42}
\ArticleNo{23}

\begin{document}

\maketitle

\begin{abstract}
We introduce and develop a set-based semantics for asynchronous $\teamltl$.  We consider two canonical logics in this setting: the extensions of $\teamltl$ by the Boolean disjunction and by the Boolean negation. We establish fascinating connections between the original semantics based on multisets and the new set-based semantics as well as show one of the first positive complexity theoretic results in the temporal team semantics setting. In particular we show that both logics enjoy normal forms that can be utilised to obtain results related to expressivity and complexity (decidability) of the new logics.  We also relate and apply our results to recently defined logics whose asynchronicity is formalized via time evaluation functions.

\end{abstract}

\section{Introduction}
\allowdisplaybreaks

Linear temporal logic (\LTL) is one of the most prominent logics for the specification
and verification of reactive and concurrent systems. The core idea in model checking, as introduced in 1977 by Amir Pnueli~\cite{pnueli}, is to
specify the correctness of a program as a set of infinite sequences, called traces, which define
the acceptable executions of the system. In $\LTL$-model checking one is concerned with trace sets that are definable by an $\LTL$-formula.
Ordinary $\LTL$ and its progeny are well suited for specification and verification of \emph{trace properties}. These are properties of systems that can be checked by going through all executions of the system in isolation. A canonical example here is \emph{termination}; a system terminates if and only if each run of the system terminates. However not all properties of interest are trace properties. Many properties that are of prime interest, e.g., in information flow security require a richer framework. 
The term \emph{hyperproperty} was coined by Clarkson and Schneider \cite{clarkson} to refer to properties of systems which relate multiple execution traces. A canonical example here would be \emph{bounded termination}; one cannot check whether a system terminates in bounded time by only checking traces in isolation.
Checking hyperproperties is vital in information flow security where dependencies between secret inputs and publicly observable outputs of a system are considered potential security violations. Commonly known properties of that type are noninterference~\cite{DBLP:conf/sp/Roscoe95,DBLP:journals/jcs/McLean92} and observational determinism~\cite{DBLP:conf/csfw/ZdancewicM03}. Hyperproperties are not limited to the area of information flow control.
E.g., distributivity and other system properties like fault tolerance can be expressed as hyperproperties~\cite{DBLP:journals/acta/FinkbeinerHLST20}.


During the past decade, the need for being able to formally specify hyperproperties has led to the creation of families of new logics for this purpose, 
since $\LTL$ and other established temporal logics can only specify trace properties.
The two main families of the new logics are the so-called \emph{hyperlogics} and logics that adopt \emph{team semantics}. In the former approach standard temporal logics such as \LTL, \CTL, and \QPTL are extended with explicit trace and path quantification, resulting in logics like \hyltl~\cite{DBLP:conf/post/ClarksonFKMRS14}, \hyctl~\cite{DBLP:conf/post/ClarksonFKMRS14}, and \HQPTL~\cite{MarkusThesis,DBLP:conf/lics/CoenenFHH19}. The latter approach (which we also adopt here) is to lift the semantics of temporal logics to sets of traces directly by adopting team semantics yielding logics such as $\teamltl$ \cite{kmvz18,GMOV22} and $\teamctl$ \cite{KrebsMV15,GMOV22}.

Krebs et al. \cite{kmvz18} introduced two versions of \LTL with  team semantics: a synchronous semantics and an asynchronous variant that differ on how the evolution of time is linked between computation traces when temporal operators are evaluated. In the synchronous semantics time proceeds in lock-step, while in the asynchronous variant time proceeds independently on each trace. For example the formula ``$\F \mathrm{terminate}$" (here $\F$ denotes the future-operator and ``$\mathrm{terminate}$" is a proposition depicting that a trace has terminated) defines the hyperproperty ``bounded termination'' under synchronous semantics, while it expresses the trace property ``termination'' under asynchronous semantics.

The elegant definition of bounded termination exemplifies one of the main distinguishing factors of team logics from hyperlogics; namely the ability to refer directly to unbounded number of traces. Each hyperlogic-formula has a fixed number of trace quantifiers that delineate the traces involved in the evaluation of the formula. Another distinguishing feature of team-based logics lies in their ability to enrich the logical language with novel atomic formulae for stating properties of teams. The most prominent of these atoms are the \emph{dependence atom} $\dep(\bar x,\bar y)$ (stating that the values of the variables $\bar x$ functionally determine the values of $\bar y$) and \emph{inclusion atom} $\bar x \subseteq \bar y$ (expressing the inclusion dependency that all the values occurring for $\bar x$ must also occur as a value for $\bar y$).

As an example, let $o_1,\dots, o_n$ be some public observables and assume that $c$ reveals some confidential information. The atom $(o_1,\dots o_n, c) \subseteq (o_1,\dots o_n, \neg c)$ expresses a form of non-inference by stating that an observer cannot infer the value of the confidential bit from the outputs.



While \hyltl and other  hyperlogics have been studied extensively, many of the basic properties of 
\teamltl are still not well understood. Krebs et al. \cite{kmvz18} showed that synchronous \teamltl and \hyltl are incomparable in expressivity and that the asynchronous variant collapses to $\LTL$\footnote{There was a slight error in the definition of asynchronous semantics in \cite{kmvz18} which we fix here.}. In this article we consider extensions of the asynchronous \teamltl, which have a greater expressive power.
Not much was know about the complexity aspects of \teamltl  until L\"uck \cite{LUCK2020} showed that the complexity of satisfiability and model checking of synchronous \teamltl with Boolean negation $\cneg$ is equivalent to the decision problem of third-order arithmetic.
Subsequently, Virtema et al. \cite{VBHKF20} embarked for a more fine-grained analysis of the complexity of synchronous \teamltl and discovered a decidable syntactic fragment (the so-called \emph{left-flat fragment}) and established that already a very weak access to the Boolean negation suffices for undecidability. They also showed that synchronous \teamltl and its extensions can be translated to \HQPTLP, which is an extension of \hyltl by (non-uniform) quantification of propositions. Kontinen and Sandstr\"om \cite{KS21} defined translations between extensions of \teamltl and the three-variable fragment of  first-order team logic to utilize the better understanding of the properties of logics in first-order team semantics in the study of \teamltl. They also showed that any logic effectively residing between  synchronous $\teamltl$ extended with the Boolean negation and second-order logic inherits the complexity properties of the extension of $\teamltl$ with the Boolean negation. 
Finally, Gutsfeld et al. \cite{GMOV22} reimagined the setting of temporal team semantics to be able to model richer forms of (a)synchronicity by developing the notion of time-evaluation functions.
In addition to reimagining the framework, they discovered decidable logics which however relied on 
restraining time-evaluation functions to be either \emph{$k$-context-bounded} or \emph{$k$-synchronous}. It is worth noting that recently asynchronous hyperlogics have been considered  also in several other articles (see, e.g.,  \cite{GutsfeldMO21,BaumeisterCBFS21}).
%


It can be said that almost all complexity theoretic results previously obtained for $\teamltl$ are negative, and the few positive results have required drastic restrictions in syntax or semantics. 
In this article we take a take a fresh look at the asynchronous variant of \teamltl.
Recent works on synchronous $\teamltl$ have revealed that quite modest extensions of synchronous $\teamltl$ are undecidable. Thus, our study of asynchronous $\teamltl$ partly stems from our desire to discover decidable but expressive logics for hyperproperties.
Until now, all the papers on temporal team semantics have explicitly or implicitly adopted a semantics based on multisets of traces. In the team semantics literature, this often carries the name \emph{strict semantics}, in contrast to \emph{lax semantics} which is de-facto set-based semantics. In database theory, it is ubiquitous that tasks that are computationally easy under set based semantics become untractable in the multiset case. In the team semantics setting this can be already seen in the model checking problem of propositional inclusion logic $\PL(\subseteq)$ which is $\Ptime$-complete under lax semantics, but $\NP$-complete under strict semantics \cite{HellaKMV19}.

Our new set-based framework offers a setting that drops the accuracy that accompanies adoption of multiset semantics, in favour of better computational properties. Consider the following formula expressing a form of strong non-inference in parallel computation:
\(
\G((o_1,...,o_n,c)\subseteq(o_1,...,o_n,\neg c)),
\)
where $o_1,...,o_n$ are observable outputs and $c$ is confidential. 
In the synchronous setting, the formula expresses that during a synchronous computation, at any given time, an observer cannot infer the value of the secret $c$ from the outputs. In the asynchronous setting, the formula states a stronger property that the above property holds for all computations (not only synchronous). In the multiset setting the number of parallel computation nodes is fixed, while in the new lax semantics, we drop that restriction, and consider an undefined number of computation nodes. The condition is stronger in lax semantics; and intuitively easier to falsify, which makes model checking in practice easier.

\textbf{Our contribution.}
We introduce and develop a set-based semantics for asynchronous TeamLTL, which we name \emph{lax semantics} and write $\teamltl^{l}$. We consider two canonical logics in this setting: the extensions of $\teamltl^{l}$ by the Boolean disjunction $\teamltl^{l}(\ovee)$ and by the Boolean negation $\teamltl^{l}(\cneg)$.
By developing the basic theory of lax asynchronous TeamLTL, we discover some fascinating connections between the strict and lax semantics.
We discover that both of the logics enjoy normal forms that can be utilised to obtain expressivity and complexity results. Tables \ref{table:exp} and \ref{table:comp} summarise our results. For comparison, Table \ref{table:comp2} summarises the known results on complexity of synchronous $\teamltl$.

\begin{table}[p!]
\begin{tabular}{ccccc}\toprule
 	$\teamltl^{s/l}$				&  		 	& 	left-flat--$\teamltl^s(\clor)$			& $\stackrel{\textrm{Cor.}\ref{cor:laxstrict}}{<}$ & $\teamltl^s(\ovee)$ \\
 \quad\rotatebox[origin=c]{-90}{$<$}{ \scriptsize Ex. \ref{ex1}}		&   	&	\rotatebox[origin=c]{-90}{$\equiv$}{\scriptsize \, Thm. \ref{lflat}}  			& & \\
 		$\teamltl^l(\ovee)$		&  	$\stackrel{\textrm{Thm. }\ref{disjnf}}{\equiv}$	 	& 	left-flat--$\teamltl^l(\clor)$	& $\stackrel{\dagger}{<}$ & quasi-flat--$\teamltl^{s/l}(\sim)$\\
%
&&&& \quad\quad\rotatebox[origin=c]{-90}{$\equiv$}{ \scriptsize Thm. \ref{qfnf}} \\
&&&& left-dc--$\teamltl^{l}(\cneg)$
\\\bottomrule\\
\end{tabular}
\caption{Expressivity hierarchy of the asynchronous logics considered in the paper. Logics with lax or strict semantics are here referred with the superscripts $l$ and $s$, respectively.
For the definitions of left flatness, quasi flatness, and left downward closure, we refer to Definitions \ref{def:leftflat} and \ref{def:quasiflat}. $\dagger$: This follows since only $\teamltl^l(\clor)$ is downward closed (cf. Theorem \ref{lflat} and Definition \ref{def:quasiflat}).
Theorem \ref{lflat} implies that for $\teamltl(\sim)$-formulae in quasi-flat form the strict and lax semantics coincide.
}
\label{table:exp}
\end{table}

\begin{table}
\begin{center}
\begin{tabular}{cccc}\toprule
Logic &  \multicolumn{2}{c}{Complexity of} & References \\ \cmidrule{2-3}
(asynchronous semantics)& model checking & satisfiability& \\ 
\midrule
$\LTL $   & \PSPACE &  \PSPACE & \cite{SC85}\\
$\PL(\cneg)$   & $\mathrm{ATIME}\mbox{-}\mathrm{ALT(exp,poly)}$ &  $\mathrm{ATIME}\mbox{-}\mathrm{ALT(exp,poly)}$ & \cite{HKVV18}\\
$\teamltl^{l/s}$   & \PSPACE &  \PSPACE & \cite{kmvz18}, Theorem \ref{thmdc}\\
left-flat-$\teamltl^{s/l}(\clor)$  & \PSPACE & \PSPACE & Theorem \ref{thm:complexity}\\
$\teamltl^l(\clor)$  & \PSPACE & \PSPACE & Theorem \ref{thm:complexity}\\
$\teamltl^s(\clor)$ & ??? & \PSPACE & $\dagger$ \\
$\teamltl^s(\dep)$ & \NEXPTIME-hard & $\PSPACE$ & \cite{kmvz18} \\
left-dc-$\teamltl^l(\cneg)$  & in $\mathrm{TOWER(poly)}$ & in $\mathrm{TOWER(poly)}$  & Theorem \ref{thm:complexity}
\\\bottomrule\\
\end{tabular}
\caption{Complexity results of this paper. All results are completeness results if not otherwise specified. $\PL(\cneg)$ refers to the propositional fragment of $\teamltl(\cneg)$ which embeds also to left-dc-$\teamltl^l(\cneg)$.
$\dagger$: All $\PSPACE$-completeness results for satisfiability in strict semantics and $\teamltl^l$ follow directly from classical $\LTL$ by downward closure and singleton equivalence similar to \cite[Proposition 5.4]{kmvz18}.
$\mathrm{ATIME}\mbox{-}\mathrm{ALT(exp,poly)}$ refers to alternating exponential time with polynomially many alternations while $\mathrm{TOWER(poly)}$ refers to problems that can be decided by a deterministic TM in time bounded by an exponential tower of 2's of polynomial height.
}
\label{table:comp}
\end{center}
\end{table}

\begin{table}
\begin{center}
\begin{tabular}{cccc}\toprule
Logic &  \multicolumn{2}{c}{Complexity of} & References \\ \cmidrule{2-3}
(sync. strict semantics)& model checking & satisfiability& \\ 
\midrule
$\teamltl$   & \PSPACE &  \PSPACE & \cite{kmvz18}\\
left-flat-$\teamltl(\clor)$  &in $\EXPSPACE$ & \PSPACE & \cite{VBHKF20}\\
$\teamltl(\dep)$ & \NEXPTIME-hard & $\PSPACE$ & \cite{kmvz18} \\
$\teamltl(\clor)$ & ??? & $\PSPACE$ & $\dagger$ \\
$\teamltl(\clor, \subseteq)$  & $\Sigma^0_1$-hard & $\Sigma^0_1$-hard & \cite{VBHKF20}\\
$\teamltl(\cneg)$  & third-order arithmetic & third-order arithmetic  & \cite{LUCK2020}
\\\bottomrule\\
\end{tabular}
\caption{Complexity results for synchronous strict semantics. All results are completeness results if not otherwise specified. $\dagger$: All $\PSPACE$-completeness results for satisfiability follow directly from classical $\LTL$ by downward closure and singleton equivalence similar to \cite[Proposition 5.4]{kmvz18}.}
\label{table:comp2}
\end{center}
\end{table}



\section{Preliminaries}

Fix a set $\ap$ of \emph{atomic propositions}. 
The set of formulae of \LTL (over $\ap$) is generated by the following grammar:
\[
\varphi \ddfn p \mid \neg p  \mid \varphi \lor \varphi \mid \varphi\land \varphi \mid \X \varphi \mid \G \varphi \mid \varphi \U \varphi , \quad\quad  \text{where $p \in \ap$.}
\]
We adopt the convention that formulae are given in negation normal form, i.e., $\neg$ is allowed to appear only in front of atomic propositions. Note that this is an expressively  complete set of $\LTL$-formulae.

A \emph{trace} $t$ over $\ap$ is an infinite sequence from $(\pow{\ap})^\omega$. For a natural number $i\in\N$, we denote by $t[i]$ the $i$th element of $t$ and by $t[i,\infty]$ the postfix $(t[j])_{j\geq i}$ of $t$. 
Semantics of \LTL is defined in the usual manner:
\begin{align*}
t &\models p                             &&\Leftrightarrow && p \in t[0] \\
t &\models \neg p                  &&\Leftrightarrow && p \not\in t[0] \\
t &\models \varphi_1 \vee \varphi_2      &&\Leftrightarrow &&t \models \varphi_1 \text{ or } t \models \varphi_2 \\
t &\models \LTLnext \varphi              &&\Leftrightarrow && t[1,\infty] \models \varphi \\
t &\models \G\varphi                     &&\Leftrightarrow && \forall i \geq 0.~ t[i,\infty] \models \varphi \\
t &\models \varphi_1 \LTLuntil \varphi_2 &&\Leftrightarrow && \exists i \geq 0.~ t[i,\infty] \models \varphi_2 \\
           &&&&&\text{ and } \forall j:0 \leq j < i \Rightarrow t[j,\infty] \models \varphi_1
\end{align*}
The \emph{truth value} of a formula $\varphi$ on
a trace $t$  is denoted by
 $\llbracket \varphi \rrbracket_{t}\in\{0,1\}$.
 
 The logical constants $\top,\bot$ and the operators $\F$ and $\W$ can be defined in the usual way: $\bot \dfn p \land \neg p$, $\top \dfn p \lor \neg p$,  $\F \phi\dfn\top\U\phi$, and $\phi\W \psi\dfn(\phi\U\psi)\vee \G \phi$.

Next we give the so-called asynchronous team semantics for \LTL introduced in \cite{kmvz18}. In \cite{kmvz18}, the release operator was defined slightly erroneously; we fix the issue here by taking $\G$ as primitive.
Informally, a multiset of traces $T$ is a collection of traces with possible repetitions. Formally, we represent $T$ as a set of pairs $(i,t)$, where $i$ is an \emph{index} (from some suitable large set) and $t$ is a trace. We stipulate that the elements of a multiset have distinct indices. From now on, we will always omit the index and write $t$ instead of $(i,t)$.
For multisets $T$ and $S$, $T \uplus S$ denotes the disjoint union of $T$ and $S$ (obtained by stipulating that traces in $S$ and $T$ have disjoint sets of indices).
Note that all the functions $f$ with domain $T$ are actually of the form $f((i,t))$ and may map different copies of the trace $t$ differently.
A \emph{team} (\emph{multiteam}, resp.) is a set (multiset, resp.) of traces.
If $f\colon T\rightarrow \mathbb{N}$ is a function, we define the update of the team through $T[f,\infty]\dfn\{t[f(t),\infty]\mid t\in T\}$, where the function $f$ determines for each trace the point in time it updates to. For functions $f$ and $f'$ as above, we write $f'<f$ if for all $t\in T$ it holds that $f'(t)<f(t)$.
The underlying team $\support(T) \dfn \{t \mid (i,t)\in T\}$ of a multiteam $T$ is called the \emph{support} of $T$.

\begin{definition}[Team Semantics for \LTL]\label{def:strict}
Let $T$ be a multiteam, $l\in \{p,\neg p \mid p\in \mathrm{AP}\}$ a literal, and $\varphi$ and $\psi$ \LTL-formulae. The asynchronous team semantics of \teamltl is defined as follows.
\begin{align*}
		&T\models l &&\Leftrightarrow&&  t\models l \text{ for all } (i,t)\in T\\ 
		&T \models\varphi\wedge\psi &&\Leftrightarrow&& T\models\varphi\text{ and }T\models\psi\\
		&T \models\varphi\vee\psi &&\Leftrightarrow&& \exists T_1,T_2 \text{ s.t. }  T_1\uplus T_2=T\text{ and }  T_1\models\varphi\text{ and }T_2\models\psi\\
		&T\models \X\varphi &&\Leftrightarrow&& T[1,\infty]\models\varphi, \text{ where $1$ is the constant function $t\mapsto 1$}\\
		&T\models \G\varphi &&\Leftrightarrow&& \forall  f\colon T\rightarrow \mathbb{N}
		  \quad T[f,\infty] \models\varphi\\
		&T\models\varphi \U\psi &&\Leftrightarrow&& \exists f\colon T\rightarrow \mathbb{N}
		  \quad T[f,\infty] \models\psi\text{ and }
		\forall f'<f \colon T'[f',\infty] \models\varphi,\\
		&&&&& \text{where } T'\dfn\{t\in T \mid f(t)\neq 0\}
\end{align*}
The synchronous variant of the semantics is obtained by allowing $f$ to range only over constant functions.
\end{definition}
We also consider the following Boolean connectives: {\em Boolean disjunction} $\clor$ and {\em Boolean negation} $\cneg$ interpreted in the usual way.
\[
\begin{array}{lcllcllcllcl}
T \models \varphi \clor \psi  &\text{ iff }& T \models \varphi \text{ or } T \models \psi \quad\quad\quad\quad\quad\quad
T \models \,\cneg \varphi  &\text{ iff }& T \not\models \varphi
\end{array}
\]
In this paper, we take the asynchronous semantics as the standard semantics and write $\teamltl$ for asynchronous $\teamltl$.

Next we define some importation semantic properties of formulae studied in the literature:
\begin{description}
\item[(Downward closure)]
If $T\models\phi$ and $S\subseteq T$, then $S\models\phi$.
\item[(Empty team property)] $\emptyset\models\phi$.
\item[(Flatness)] $T\models\phi$ ~~iff~~ $\{t\} \models\phi$ for all $t\in T$.
\item[(Singleton equivalence)] $\{t\}\models\phi$ ~~iff~~ $t\models\phi$.
\end{description}
\noindent A logic has one of the above properties if every formula of the logic has the property.
It is easy to check that \teamltl has all the above properties \cite{kmvz18}.

We will now justify our choice of semantics. The semantic rules for literals, conjunction, and disjunction are the standard ones in team semantics, and which have been motivated numerous times in the literature \cite{vaananen07}. The two main desirable properties for the logic to have are flatness and singleton equivalence, which also motivated the original definition of asynchronous $\teamltl$ \cite{kmvz18}. The given semantics for $\X$ is the only possible one that satisfies flatness. The same is true for $\F$ (i.e., $\top \U \varphi$) and $\G$; moreover the semantics clearly capture the intuitive meanings of asynchronously in the future and asynchronously globally, respectively. The given semantics for $\U$ preserves flatness and singleton equivalence, and adequately captures the intuitive meaning of asynchronous until.
The framework of asynchronous $\teamltl$ then allows us to define different variants of the familiar temporal operators. E.g.,
$\varphi \W_1 \psi \dfn \G \varphi \lor \varphi \U \psi$ and $\varphi \W_2 \psi \dfn \G \varphi \clor \varphi \U \psi$ define different variants of \emph{weak until}; the first of which is flat, while the second is not.
\begin{align*}
&T\models\varphi\W_1\psi &&\Leftrightarrow&& \exists T_1,T_2 \text{ s.t. }  T_1\uplus T_2=T,\,   T_1\models\G \varphi \text{ and }T_2\models \varphi\U\psi \\
&T\models\varphi\W_2\psi &&\Leftrightarrow&&  T\models\G \varphi \text{ or }T\models \varphi\U\psi
\end{align*}
Similarly $\varphi \R_1 \psi \dfn \psi\U((\psi\wedge\varphi)\vee\G\psi)$ and $\varphi \R_2 \psi \dfn \psi \U ((\psi\wedge\varphi)\clor\G\psi)$ give rise to different variants of \emph{release}.
Moreover, with $\cneg$ one can define additional dual operators.

A defining feature of team semantics is the ability to enrich logics with novel atomic statements describing properties of teams in a modular fashion. For example,
{\em dependence atoms} $\dep(\phi_1,\dots,\phi_n,\psi)$ and {\em inclusion atoms} $\phi_1,\dots,\phi_n  \subseteq  \psi_1,\dots,\psi_n$, with $\phi_1,\dots,\phi_n ,\psi,\psi_1,\dots,\psi_n$ being $\LTL$-formulae, have been studied extensively in first-order and modal team semantics. The dependence atom states that the truth value of $\psi$ is functionally determined by that of $\phi_1, \ldots, \phi_n$ whereas the inclusion atom states that each value combination of $\phi_1, \ldots, \phi_n$ must also occur as a value combination for $\psi_1, \ldots, \psi_n$. The  semantics of these  atoms can be given as follows:
\begin{eqnarray*}
T\models  \dep(\phi_1,\dots,\phi_n,\psi)  ~~\text{iff}~~ \forall t ,t' \in T :
\Big( \bigwedge_{1\leq j\leq n}  \truth{\phi_j}_{t} = \truth{\phi_j}_{t'} \Big) \Rightarrow \truth{\psi}_{t} = \truth{\psi}_{t'}
\end{eqnarray*}
\[
T\models \phi_1,\dots,\phi_n  \subseteq  \psi_1,\dots,\psi_n  ~~\text{iff}~~ \forall t\in T \, \exists t'\in T :
\bigwedge_{1\leq j\leq n}  \truth{\phi_j}_{t} = \truth{\psi_j}_{t'}		
\]
Consider the following example formula with dependence atoms:
	\[
	\G \dep(i_1, i_2, o) \lor \G \dep(i_2, i_3, o)
	\]
This formula states that the executions of the system can be decomposed into two parts; in the first part, the output $o$ is determined by the inputs $i_1$ and $i_2$, and in the second part, $o$ is determined by the inputs $i_2$ and $i_3$.

If $\mathcal A$ is a collection of atoms and connectives, $\teamltl(\mathcal A)$ denotes the extension of $\teamltl$ with the atoms and connectives in $\mathcal A$. 
It is straightforward to see (in analogy  to the modal team semantics setting \cite{KontinenMSV14}) that any dependency such as the ones above is determined by a finite set of $n$-ary Boolean relations. Let $B$ be a set of $n$-ary Boolean relations. We define the property $[\varphi_1,\dots,\varphi_n]_B$ for an $n$-tuple $(\varphi_1,\dots,\varphi_n)$ of $\LTL$-formulae:
\[
T \models [\varphi_1,\dots,\varphi_n]_B \quad\text{iff}\quad \{ (\truth{\phi_1}_{t}, \dots,  \truth{\phi_n}_{t} ) \mid t\in T  \} \in B.
\]
Expressions of the form $[\varphi_1,\dots,\varphi_n]_B$ are \emph{generalised atoms}.
It was shown in \cite{VBHKF20} that, in the synchronous setting, $\teamltl(\sim)$ is expressively complete with respect to all generalised atoms, whereas the extension of $\teamltl(\ovee)$
with the so-called \emph{flattening operator}
can express any downwards closed generalised atoms.
These results readily extend to the asynchronous setting. Moreover the flattening operator renders itself unnecessary due to flatness of asynchronous $\teamltl$.
The results imply, e.g, that the (downwards closed) dependence atoms 
can be expressed in both of the logics $\teamltl(\sim)$ and $\teamltl(\ovee)$, and inclusion atoms in turn are expressible in $\teamltl(\sim)$. The proof of the following theorem is essentially the same as the proof of \cite[Proposition 17]{DBLP:journals/corr/abs-2010-03311}.
\begin{theorem}
Let $\mathcal{A}$, $\mathcal{D}$ be the sets of all generalised atoms, and all downward closed generalised atoms. Then $\teamltl(\mathcal{A},\ovee) \equiv \teamltl(\ovee)$ and $\teamltl(\mathcal{D},\sim) \equiv \teamltl(\sim)$.
\end{theorem}
In the above $L \equiv L'$ denotes the equivalence of the logics $L$ and $L'$. 

 \section{Set-based semantics for TeamLTL}

Next we define a relaxed version of the asynchronous semantics. It is called \emph{lax} semantics as it corresponds to the so-called lax semantics of first-order team semantics (see e.g., \cite{galliani12}). From now on we refer to the semantics of Definition \ref{def:strict} as strict semantics. The possibility of considering lax semantics for  $\teamltl$ was suggested by L\"uck already in \cite{LuckPHD20} but the full definition was not given. Intuitively, lax semantics can always be obtained from a strict one by checking what strict semantics would yield if multiteams would be enriched with unbounded many copies of each of its traces. 
One of the defining features of lax semantics is that it unable to distinguish multiplicities, which is formalised by
Proposition \ref{prop:setbased} below.

We need some notation for the new definition. We write $\mathcal{P}(\mathbb{N})^+$ to denote $\mathcal{P}(\mathbb{N})\setminus \{\emptyset\}$. For a team $T$ and function $f\colon T\rightarrow \mathcal{P}(\mathbb{N})^+$, we write $T[f,\infty]$ to denote the team $\{t[s,\infty]\mid t\in T,s\in f(t)\}$.
For $T'\subseteq T$, $f \colon T \to \mathcal{P}(\mathbb{N})^+$, and $f' \colon T' \to \mathcal{P}(\mathbb{N})^+$, we define that $f'< f$ if and only if
\[
\text{$\forall t \in T'$: $\min(f'(t)) \leq \min(f(t))$ and $\max(f'(t)) < \max(f(t))$ (assuming $\max(f(t))$ exists)}.
\]
Now for $f \neq f'$, it holds that $f' < f$ if and only if there exists a function $g\colon T'[f',\infty] \to \mathcal{P}(\mathbb{N})^+$ such that $\big(T'[f',\infty]\big)[g,\infty] = T'[f,\infty]$.

  \begin{definition}[$\teamltl^{l}$]
    Suppose $T$ is a team, and $\varphi$ and $\psi$ are \teamltl formulae. The lax semantics is defined as follows. We only list the cases that that differ from the standard semantics.
      \begin{align*}
     	&T \models^l\varphi\vee\psi &&\Leftrightarrow&& \exists T_1,T_2 \text{ s.t. }  T_1\cup T_2=T\text{ and }  T_1\models\varphi\text{ and }T_2\models\psi\\
			&T\models^l \G\varphi &&\Leftrightarrow&& \forall  f\colon T\rightarrow \mathcal{P}(\mathbb{N})^+ 
		 \text{ it holds that } T[f,\infty] \models^l\varphi \\
			&T\models^l\varphi \U\psi &&\Leftrightarrow&& \exists f\colon T\rightarrow \mathcal{P}(\mathbb{N})^+  \text{ such that } T[f,\infty] \models^l\psi\mbox{ and } \\
			&&&&& \forall f'\colon T'\rightarrow \mathcal{P}(\mathbb{N})^+ \text{s.t. $f' < f$, it holds that } T'[f',\infty] \models^l\varphi \text{ or } T'= \emptyset, \\
			&&&&& \text{where }T'\dfn\{t \in T \mid \max(f(t))\neq 0  \}	
      \end{align*}
  \end{definition}
  In the context we will be considering in this article, the subformulae $\varphi$ in the definition of the until operator $\U$ always has the empty team property and thus we disregard the possibility that the team $T'$ is empty in our proofs, as that case follows from the empty team property.
  
  The above set-based semantics can also be viewed in terms of multisets. In that case functions $f$ are quantified uniformly, i.e. we restrict our consideration to functions where $f(i,t) = f(j,t)$. Furthermore semantics for disjunctions is defined in a way that omit references to multiplicities. We extend the lax semantics to multiteams $T$ by stipulating that $T\models^l \varphi$ iff $\support(T)\models^l \varphi$. This stipulation makes it easier to relate our new logic to the old multiteam based ones.

The following theorem shows that $\teamltl^{l}(\cneg)$ satisfies the so-called \emph{locality} property, see Appendix \ref{a:one} for the proof.
For a trace $t$ over $\ap'$ and $\ap\subseteq \ap'$, the \textit{reduction of $t$} to $\ap$, $t_{\upharpoonright\ap}$, is a sequence from $(2^\ap)^\omega$ such that $p\in t[i]$ if and only if $p\in t_{\upharpoonright\ap}[i]$, for all $p\in\ap$ and $i\in\N$. For a team $T$ over $\ap'$ we define the \textit{reduction of $T$} to $\ap$ by $T_{\upharpoonright\ap} = \{t_{\upharpoonright\ap} \mid t\in T\}$.
\begin{restatable}{proposition}{setbased}\label{prop:setbased}
Let $T$ be a team and $\varphi$ a $\teamltl^{l}(\cneg)$-formula with variables in $\ap$. Now $T\models^l \varphi$ iff $T_{\upharpoonright \ap}\models^l \varphi$.
\end{restatable}

 The next theorem displays that lax semantics enjoys the same fundamental properties as its strict counterpart. The proof via a straightforward induction, see Appendix \ref{a:one} for details.

 \begin{restatable}{theorem}{basicprop}\label{thmdc}
 $\lteamltl$ satisfies downward closure, empty team property, singleton equivalence, and flatness.
 \end{restatable}

 The following example establishes that the new lax semantics differs from the strict semantics, and that in the old semantics multiplicities matter.
 Moreover, we obtain $\teamltl^l<\teamltl^l(\varovee)$ by showing that the latter is not flat.
 \begin{example}\label{ex1} Let $\phi$ be the formula $\G(p\ovee q)$, 
  $T_1 \dfn \{t\}$ and $T_2 \dfn \{(1,t),(2,t)\}$, where $t \dfn \{p\} \{q\}^{\omega}$. It is easy to check that $T_1\models \phi$ but 
   $T_1\not\models^l \phi$ (which is witnessed by $T[f,\infty]\not\models^l p\ovee q$ for $f(t) \dfn \{0,1\}$). Likewise $T_2\not\models \phi$. 
 Moreover $\{s_i\} \models^l \phi$, for $i\in\{1,2\}$, but $\{s_1,s_2\} \not\models^l \phi$, where  $s_1 \dfn \{p\}^{\omega}$ and $s_2 \dfn \{q\}^{\omega}$.
 \end{example}
 
We will also consider the following fragments of $\teamltl(\clor)$ and  $\teamltl(\sim)$.

\begin{definition}\label{def:leftflat}
A formula $\phi$ of  $\teamltl(\clor)$ is called \emph{left-flat}, if in  all of its subformulae of  the form $\G\psi$ and  $\psi\U\theta$, the subformula $\psi$ is an \LTL-formula. 
A formula $\phi$ of  $\teamltl(\sim, \ovee)$ is called \emph{left-downward closed}, if in  all of its subformulae of  the form $\G\psi$ and  $\psi\U\theta$, the subformula $\psi$ is an $\teamltl(\clor)$-formula. 
\end{definition}

We will later show that the above syntactic restriction for flatness could be replaced by a semantic restriction as well (see Corollary \ref{cor:flatness}). The proof of the following Theorem can be found in Appendix \ref{a:one}.
 
 \begin{restatable}{theorem}{leftflat}\label{lflat}
 For all $\varphi \in \teamltl^l(\clor)$ the following two claims hold:
\begin{enumerate}
\item\label{lflat_i} $\varphi$ is downward closed, and
\item\label{lflat_ii} if $\phi$ is left-flat, then $T\models \phi$ iff $\support(T)\models^l \phi$ for all multiteams $T$.
 \end{enumerate}
 \end{restatable}
The restriction to left-flat formulae  in case \eqref{ii} above is necessary  by Example \ref{ex1}.

 \section{Normal Form Theorems}
 In this section we develop normal forms for our logics, which we then utilise to obtain strong expressivity and complexity results.

 \subsection{TeamLTL with the Boolean disjunction}
 
  \begin{definition}A formula $\phi$ is in $\ovee$-disjunctive normal form if it is of the form
 \(
 \Clor_{i\in I} \alpha_i, 
 \)
 where $\alpha_i$ are \LTL-formulae.
\end{definition}
We will next show that every formula of   $\teamltl^{l}(\ovee)$ can be transformed into equivalent $\ovee$-disjunctive normal form. This result is similar to the one proved in \cite{Virtema17} for team-based modal logic $\ML(\ovee)$.

\begin{theorem}\label{disjnf} Every formula  $\phi\in \teamltl^{l}(\ovee)$ is logically equivalent to a formula $\phi^*= \Clor_{i\in I} \alpha_i$ in $\ovee$-disjunctive normal form, where $\lvert \alpha_i \rvert \leq \lvert \phi \rvert$ and  $\lvert I \rvert = 2^k$,  where $k$ is the number of $\clor$ in $\varphi$. 
 \end{theorem}
 \begin{proof}
 The proof proceeds by induction on the structure of formulae. Note that  atomic formulae are already in the normal form and that the case for $\ovee$ is trivial. The remaining cases are defined as follows:
  \begin{align*}
    (\psi \wedge \theta)^*  \dfn& \Clor_{i\in I ,j\in J} (\alpha^\psi_i \wedge \alpha^\theta_j) &&&
    (\psi \vee \theta)^*    \dfn& \Clor_{i\in I ,j\in J}  (\alpha^\psi_i \vee \alpha^\theta_j) \\
    (\X\psi)^*              \dfn& \Clor_{i\in I} \X\alpha^\psi_i  &&& 
    (\G\psi)^*              \dfn& \Clor_{i\in I} \G\alpha^\psi_i \\
    (\psi \U \theta)^*      \dfn& \Clor_{i\in I,j\in J} (\alpha^\psi_i \U \alpha^\theta_j). &&&&
  \end{align*}
where $\alpha^\psi_i$ and $\alpha^\theta_j$ are the flat formulae in the disjunctive normal forms of $\psi$ and $\theta$ respectively, and $I$ and $J$ are the respective index sets.

 Suppose $\phi = \psi \wedge \theta$ and that $\psi\equiv \Clor_{i\in I} \alpha^{\psi}_i$ and $\theta\equiv  \Clor_{i\in J} \alpha^{\theta}_j$  (induction hypothesis). Now $T \models^l \phi$ if and only if $T \models^l \psi$ and $T \models^l \theta$.  
 The latter holds, if and only if $T \models^l \alpha^\psi_k$ and $T \models^l \alpha^\theta_l$, for some $k$ and $l$. This can be equivalently expressed as $T \models^l \Clor_{i,j} (\alpha^\psi_i \wedge \alpha^\theta_j)$, i.e. $T \models^l \phi^*$.
   
   Suppose $\phi = \psi \vee \theta$ and  that $\psi\equiv \Clor_{i\in I} \alpha^{\psi}_i$ and $\theta\equiv  \Clor_{i\in J} \alpha^{\theta}_j$. By definition $T \models^l \phi$ if and only if there exists $T'\cup T'' = T$ such that $T'\models^l \psi$ and $T'' \models^l \theta$. By the induction hypothesis the latter is equivalent with $T' \models^l \Clor_{i\in I} \alpha^\psi_i$ and $T'' \models^l \Clor_{j\in J} \alpha^\theta_j$. By definition this holds if and only if there are $k$ and $l$ such that $T' \models^l \alpha^\psi_k$ and $T'' \models^l \alpha^\theta_l$, which is equivalent with $T \models^l \alpha^\psi_k \vee \alpha^\theta_l$, by definition. Equivalently then $T \models^l \Clor_{i\in I,j\in J} (\alpha^\psi_i \vee \alpha^\theta_j)$.
   
     Suppose $\phi = \X\psi$ and  that $\psi\equiv \Clor_{i\in I} \alpha^{\psi}_i$. By definition $T \models^l \phi$ is equivalent with $T[1,\infty] \models^l \psi$. By the induction hypothesis the latter holds if and only if $T[1,\infty] \models^l\Clor_{i\in I} \alpha^{\psi}_i$, which by definition is equivalent with $T[1,\infty] \models^l \alpha^\psi_k$ for some $k \in I$. The latter holds if and only if $T \models^l \X \alpha^\psi_k$ for some $k\in I$, which is equivalent with $T \models^l \Clor_{i\in I} \X \alpha^\psi_i$.

     Suppose $\phi = \G\psi$ and  that $\psi\equiv \Clor_{i\in I} \alpha^{\psi}_i$. Suppose that $T \models^l \phi$. By definition for all functions $f\colon T\to \mathcal{P}(\mathbb{N})^+$ it holds that $T[f,\infty] \models^l \psi$. By the induction hypothesis $T[f,\infty] \models^l \Clor_{i\in I} \alpha^\psi_i$ for all $f$. Especially this holds for the total function defined for every $t\in T$ by  $f_{max}(t) \dfn \N$.  
     Thus $T[f_{max},\infty] \models^l \alpha^\psi_k$ for some $k$. By downward closure it holds that $T[f',\infty] \models^l \alpha^\psi_k$ for all $f' \leq f_{max}$. Hence $T \models^l \G\alpha^\psi_k$, and thus $T \models^l \Clor_{i\in I} \G\alpha^\psi_i$. The other direction is analogous.
     
     Suppose $\phi = \psi \U \theta$ and that  $\psi\equiv \Clor_{i\in I} \alpha^{\psi}_i$   and $\theta\equiv  \Clor_{j\in J} \alpha^{\theta}_j$. Suppose $T \models^l \phi$. By definition there exists a function $f\colon T\rightarrow \mathcal{P}(\mathbb{N})^+$ such that $T[f,\infty] \models^{l} \theta$ and for all functions  $f'\colon T'\rightarrow \mathcal{P}(\mathbb{N})^+$ such that $f'< f$, $T'[f',\infty] \models^{l} \psi$, where $T' \dfn \{t\in T \mid f(t) \neq 0\}$.
  Hence by the induction hypothesis $T[f,\infty] \models^l \Clor_{j\in J} \alpha^{\psi}_j$, which is equivalent with $T[f,\infty] \models^l \alpha^\psi_k$ for some $k \in J$, and, for the function $f_{max}(t) \dfn \{ n \in \N \mid n < m, \textrm{ for some }m\in f(t) \}$, it holds that $T[f_{max},\infty] \models^l \Clor_{i\in I} \alpha^{\theta}_i$, which in turn is equivalent with $T[f_{max},\infty] \models^l \alpha^\psi_l$ for some $l \in I$. By downward closure the latter holds for all intermediary functions, and thus $T \models^l \alpha^\theta_k\U\alpha^\psi_l$ and finally 
 $T \models^l \Clor_{i\in I,j\in J} (\alpha^\theta_i \U \alpha^\psi_j)$ as wanted. The other direction is analogous.

In order to show the size estimates stated in the theorem, it suffices to note that our translation to  $\ovee$-disjunctive normal from can be equivalently stated as follows: 
\[\phi\equiv \Clor_{i\in I} \alpha^{\psi}_i= \Clor_{f\in F} \phi^f,\]
where  $F$ is  the set of all selection functions $f$ that select, separately for each occurrence, either the left disjunct $\psi$ or the
right disjunct $\theta$ of each subformula of the form $\psi\ovee \theta$ of $\phi$, and $\phi^f$ denotes the formula obtained from $\phi$ by substituting
each occurrence of a subformula of type $(\psi\ovee \theta)$ by $f(\psi\ovee \theta)$. The size estimates follow immediately from this observation.
 \end{proof}

Proofs for the following two corollaries can be found in Appendix \ref{a:one}.

\begin{restatable}{corollary}{corflatness}\label{cor:flatness}
For every flat $\teamltl^l(\clor)$-formula there exists an equivalent $\teamltl^l$-formula.
\end{restatable}

\begin{restatable}{corollary}{corlaxstrict}\label{cor:laxstrict}
$\teamltl^{l}(\clor) < \teamltl(\clor)$
\end{restatable}

\subsection{TeamLTL with the Boolean negation}
 
A normal form, similar to the one in Theorem \ref{disjnf}, can also be obtained for \teamltl extended with the Boolean negation, however, since this extension is no longer downward closed, it only holds for a specific fragment of the logic.
 The following normal form has been introduced and used in \cite{Luck18a,Luck18b} to analyse the complexity of modal team logic and $\FO^2$ in the team semantics context.
 Below  $\phi^d$  denotes  a formula obtained by transforming $\neg \phi$ into negation normal form in the standard way in \LTL.

 \begin{definition}\label{def:quasiflat}
 A formula $\phi$ is quasi-flat if $\phi$ is of the form:
 \[ \Clor_{i\in I} (\alpha_i\wedge \bigwedge_{j\in J_i} \exists \beta_{i,j}),  \]
where $\alpha_i$ and  $\beta_{i,j}$ are \LTL-formulae, and $\exists \beta_{i,j}$ is an abbreviation for the formula $\sim \beta^d_{i,j}$.
\end{definition} 
Note that, for \LTL-formulae $\alpha$ and $\beta$, we have $T\models^l \alpha$ if and only if $t\models \alpha$, for all $t\in T$. Moreover $T\models^l \exists \beta$, if and only if there exists some trace $t\in T$ such that $t \models \beta$.
 
 \begin{theorem}\label{qfnf} Every left-downward closed formula $\phi\in \teamltl^{l}(\sim, \ovee)$ is logically equivalent to a quasi-flat formula $\phi^*$. 
 \end{theorem}
 \begin{proof}
   Proof by induction over the structure of $\phi$.      Atoms are flat, and hence are in the normal form. The translations and the proofs of correctness for the cases of conjunction, disjunction, and Boolean negation are analogous to the simpler modal framework of  \cite{Luck18a,Luck18b}. 
     
     Suppose $\varphi = \psi \wedge \theta$ and assume that $\psi$ is equivalent to 
     $ \Clor_{i\in I} (\alpha^{\psi}_i\wedge \bigwedge_{j\in J_i} \exists \beta^{\psi}_{i,j})$ and $\theta$ to 
    $   \Clor_{i\in I'} (\alpha^{\theta}_i\wedge \bigwedge_{j\in J'_i} \exists \beta^{\theta}_{i,j})$. By the distributive laws of conjunction and disjunction,  $\phi $ is clearly equivalent to \[\Clor_{i\in I,k\in I'} (\alpha^\psi_i\wedge\alpha^\theta_k\wedge\bigwedge_{j\in J_i} \exists \beta_{i,j}\wedge\bigwedge_{j\in J'_k} \exists \beta_{k,j}).\]

     Suppose $\varphi = \psi \vee \theta$. By the induction hypothesis and an argument analogous to the disjunction case of the proof of Theorem \ref{disjnf}, $\phi$ is equivalent to 
     \begin{align}\label{qfnf:dsj}
     \Clor_{i\in I,k\in I'} \big((\alpha^\psi_i\wedge\bigwedge_{j\in J_i} \exists \beta^\psi_{i,j}) \vee ( \alpha^\theta_k\wedge \bigwedge_{j\in J'_j} \exists \beta^\theta_{k,j})\big).
     \end{align}
      The above formula expresses that $T$ can be split into two parts: $T_1$ in which each trace satisfies $\alpha_i$ and the subformulae $\beta_{i,j}$ are satisfied by some traces, and $T_2$ in which each trace satisfies $\alpha_k$ and the subformulae $\beta_{k,j}$ are satisfied by some traces. But this is equivalent to saying that $T$ can be split into two parts: $T_1$ in which each trace satisfies $\alpha_i$, and $T_2$ in which each trace satisfies $\alpha_k$; and the subformulae $\alpha_i \wedge \beta_{i,j}$ and $\alpha_k \wedge \beta_{k,j}$ are satisfied by some traces in $T$, and thus the formula \eqref{qfnf:dsj} is equivalent with
      \[\Clor_{i\in I,k\in I'} \big((\alpha^\psi_i\vee\alpha^\theta_k)\wedge \bigwedge_{j \in J_i} \exists(\alpha^\psi_i\wedge \beta^\psi_{i,j})\wedge \bigwedge_{j \in J'_k} \exists(\alpha^\theta_j\wedge \beta^\theta_{k,j})\big)\]
       that is in the normal form.

    Suppose $\varphi = \sim\psi$ and assume that $\psi$ is equivalent to 
     $ \Clor_{i\in I} (\alpha_i\wedge \bigwedge_{j\in J_i} \exists \beta_{i,j}).$  Now $\varphi$ is clearly  equivalent to 
     \(
     \bigwedge_{i\in I} (\exists\alpha^d_i\clor \Clor_{j\in J_i} \beta^d_{i,j}).
     \)
     This formula can be expanded  back to the  normal form with exponential blow-up using  the distributivity law of propositional logic.
     
     Suppose $\varphi = \X\psi$ and assume that $\psi$ is equivalent to 
     $ \Clor_{i\in I} (\alpha_i\wedge \bigwedge_{j\in J_i} \exists \beta_{i,j})$.  It is now easy to check that  $\varphi$ is equivalent to $\Clor_{i\in I}( \X\alpha_i\wedge \bigwedge_{j\in J_i} \exists\X\beta_{i,j})$. 
    
    Suppose  $\varphi = \G\psi$. Since $\varphi$ was assumed to be left-downward closed,  $\psi$ is equivalent with a formula of the form $\ovee_i\alpha$. Now we can transform $\varphi$ to the normal form as in the previous theorem.   
   
Suppose $\varphi = \psi\U\theta$. By assumption $\varphi $ is left-downward closed hence $\psi$ is equivalent with a formula of the form $\ovee_{i\in I} \alpha^\psi_i$ (by the previous theorem) and $\theta$ is equivalent to 
     $ \Clor_{k\in I'} (\alpha^{\theta}_k\wedge \bigwedge_{j\in J_k} \exists \beta^{\theta}_{k,j})$. Now using the fact that $\psi$ is downward closed,  it is easy to see that  $\varphi$ is logically equivalent with the formula:
    \begin{align}\label{eqnf1}
        \Clor_{i\in I, k\in I'} \big(\alpha^{\psi}_i \U (\alpha^{\theta}_k\wedge \bigwedge_{j\in J_k} \exists \beta^{\psi}_{k,j})\big).  
    \end{align}
 It now suffices to show that the disjuncts (for any $i\in I,k\in I'$) of \eqref{eqnf1} can be equivalently expressed as:
\begin{align}\label{eqnf2}
  \big(\alpha^{\psi}_i\U\alpha^{\theta}_k\wedge  \bigwedge_{j\in J_k}\exists(\alpha^{\psi}_i\U(\alpha^{\theta}_k\wedge  \beta^{\psi}_{k,j})\big).   
  \end{align}
  We will show the logical implication from \eqref{eqnf2} to \eqref{eqnf1}. Assume 
  \[T\models^l    \big(\alpha^{\psi}_i\U\alpha^{\theta}_k\wedge  \bigwedge_{j\in J_k} \exists( \alpha^{\psi}_i\U(\alpha^{\theta}_k\wedge  \beta^{\psi}_{k,j})\big).     \]   
  Let $f$ be such that $T[f,\infty]\models^l   \alpha^{\theta}_k$ and that  $T[g,\infty]\models^l   \alpha^{\psi}_i$, for all $g<f$. In order to show  
   \begin{equation}\label{eqnf3}
   T\models^l  \alpha^{\psi}_i \U (\alpha^{\theta}_k\wedge \bigwedge_{j\in J_k} \exists \beta^{\psi}_{k,j}),  
   \end{equation}
we need to make sure that traces witnessing the truth of the formulae $\exists \beta^{\psi}_{k,j}$ can be found in $T[f,\infty]$. Here we can use the assumption that $T\models^l  \bigwedge_{j\in J_k} \exists( \alpha^{\psi}_i\U(\alpha^{\theta}_k\wedge  \beta^{\psi}_{k,j}))$
 implying that for each $j\in J_k$ there exists $t_j\in T$ such that $t_j\models  \alpha^{\psi}_i\U(\alpha^{\theta}_k\wedge  \beta^{\psi}_{k,j})$. Let now  $n_j$ be  such that 
  $t_j[n_j,\infty]\models \alpha^{\theta}_k\wedge  \beta^{\psi}_{k,j} $ and that   $t_j[l,\infty]\models \alpha^{\psi}_i$ for all $l<n_j$. Now by the flatness of the formulae $\alpha^{\psi}_i, \alpha^{\theta} _k$, and $\beta^{\psi}_{k,j}$, the function $f'$ defined by
 \begin{equation*}
    f'(t) \dfn
    \begin{cases*}
       f(t)\cup\{t_j[n_j, \infty]\} & if $t=t_j$, for some $j\in J_k$  \\
         f(t)    & otherwise
    \end{cases*}
  \end{equation*}
 witnesses \eqref{eqnf3}. The converse is proved analogously.
  \end{proof}
 
The following example indicates that the restriction to left-downward closed formulae is necessary for the proof to work in the above theorem.

 \begin{example}\label{ex2} Let $\phi$ be the formula $\G(\exists p_1\ovee \exists p_2)$ 
  and $T\dfn\{t\}$, where $ t\dfn (\{p_1\}\{p_2\} )^{\omega} $. It is now easy to check that $T\models^l \phi$ but 
   $T\not\models^l \G\exists p_i$ for $i\in \{1,2\}$.

 \end{example}

\section{Computational Properties}
In this section we analyse the computational properties of the logics studied in the previous section. 
We  focus on the complexity of the model checking and satisfiability problems.

For the model checking problem one has to determine  whether a  team of traces generated by a given finite Kripke structure satisfies a given formula. We consider Kripke structures of the form $\kK \dfn (W, R, \eta , w_0)$, where $W$ is a finite set of states, $R\subseteq W^2$ a left-total transition relation, $\eta\colon W\rightarrow 2^{\mathrm{AP}}$ a labelling function, and $w_0\in W$ an initial state of $W$. A path $\sigma$ through $\kK$ is an infinite sequence $\sigma \in W^\omega$ such that $\sigma[0] \dfn w_0$ and $(\sigma[i], \sigma[i + 1]) \in R$ for every $i \geq 0$. The \emph{trace of $\sigma$} is defined as $t(\sigma) \dfn \eta(\sigma[0])\eta(\sigma[1])\dots \in (2^\mathrm{AP})^\omega$.  A Kripke structure $\kK$ then \emph{generates} the trace set  $\traces(\kK) \dfn \{t(\sigma) \mid \sigma \text{ is a path through $\kK$}\}$.
 
\begin{definition}

The \emph{model checking problem} of a logic $\LL$ is the following decision problem:
Given a formula $\varphi\in\LL$  and a Kripke structure $K$ over $\mathrm{AP}$, determine whether $\mathit{Traces}(K) \models \varphi$,
The \emph{(countable) satisfiability problem} of a logic $\LL$ is the following decision problem:
Given a formula $\varphi\in\LL$, determine whether $T\models \phi$ for some (countable)   $T\neq \emptyset$.
\end{definition}

 Below we will use the fact that the model checking and satisfiability problems of \LTL are $\PSPACE$-complete \cite{SC85}. Furthermore, we use the facts that  the satisfiability problem of propositional team logic $\mathrm{PL}(\sim)$ is $\mathrm{ATIME}\mbox{-}\mathrm{ALT(exp,poly)}$-complete \cite{HKVV18}, and that the complexity of modal team logic is complete for the class  $\mathrm{TOWER(poly)} \dfn \mathrm{TIME}(\exp_{n^{O(1)}}(1))$, where $\exp_{0}(1) \dfn 1$ and $\exp_{k+1}(1) \dfn 2^{\exp_{k}(1)}$ \cite{Luck18a,Luck18b}.

\begin{theorem}\label{thm:complexity}
\begin{enumerate}
\item\label{i}  The model checking and satisfiability problems of $\teamltl^l(\clor)$ are $\PSPACE$-complete.
\item\label{ii} The model checking and satisfiability problems of the left-flat fragment of  $\teamltl(\clor)$ are  $\PSPACE$-complete.
\item\label{iii}  The model checking  problem of the left-downward closed fragment of $\teamltl^l(\sim,\clor)$ is $\PSPACE$-hard and it is contained in $\mathrm{TOWER(poly)}$.
\item\label{iv}  The satisfiability problem of the left-downward closed fragment of $\teamltl^l(\sim,\clor)$ is $\mathrm{ATIME}\mbox{-}\mathrm{ALT(exp,poly)}$-hard and it is contained in  $\mathrm{TOWER(poly)}$.
\end{enumerate}
\end{theorem} 

\begin{proof}
Let us first consider the proofs of claims \ref{i} and \ref{ii}.  Note that \PSPACE-hardness holds already for \LTL-formulae, hence it suffices to show containment in \PSPACE. Furthermore, note that  \ref{ii} follows immediately from \ref{i} and Theorem \ref{lflat}.  Assume a formula $\phi\in \teamltl^l(\clor)$  and a Kripke structure $K$ is given as input. By  Theorem \ref{disjnf}, $\phi$ is logically equivalent with a formula of the form  $ \Clor_{f\in F} \phi^f$, where  $f$ varies over selection functions  selecting, separately for each occurrence, either the left disjunct $\psi$ or the
right disjunct $\theta$ of each subformula of the form $\psi\ovee \theta$ of $\phi$.
Now, without constructing the full formula $\Clor_{f\in F} \phi^f$, using polynomial space with respect to the size of $\phi$ it is possible to check whether $\traces(\kK) \models \phi_f$ for some $f\in F$. Hence the upper bound follows from the fact that \LTL model checking is in \PSPACE. The upper bound for satisfiability follows analogously.

Let us then consider the proof of claim  \eqref{iv}. The proof of claim  \eqref{iii} is analogous. 
For the lower bound it suffices to note that  propositional team logic $\mathrm{PL}(\sim)$ is a fragment of the left-downward closed fragment of $\teamltl^l(\sim,\clor)$ and hence its satisfiability problem can be trivially reduced  to the  satisfiability problem of the  left-downward closed fragment. Therefore  $\mathrm{ATIME}\mbox{-}\mathrm{ALT(exp,poly)}$-hardness follows by the result of  \cite{HKVV18}.

For the upper bound we first  transform an input formula $\phi $ into an equivalent quasi-flat formula of the form 
 \begin{equation}\label{qfnfc}
  \Clor_{i\in I} (\alpha_i\wedge \bigwedge_{j\in J_i} \exists \beta_{i,j}).  
  \end{equation}
Analogously to  \cite{Luck18a,Luck18b}, this formula can be computed in time  $\mathrm{TIME}(\exp_{O(|\phi |)}(1))$. It is now easy to see that the formula  \eqref{qfnfc} is satisfiable iff there exists  $i\in I$, such that 
$\mathrm{SAT}(\alpha_i \wedge \beta_{i,j})=1$  for all $j\in J_i$. 
Since \LTL-satisfiability checking is contained in $\PSPACE\subseteq \mathrm{TIME}(2^{n^{O(1)}})$, the overall complexity of the above procedure is in 
$\mathrm{TIME}(\exp_{(|\phi |^{O(1)})}(1))$.
\end{proof}

\section{Connections to Other Forms of Asynchronicity}

In \cite{GMOV22} the authors introduced a novel team-based logic that can deal with different modes of asynchronous hyperproperties by using so-called \emph{time evaluation functions} (\tefs).
Time evaluation functions facilitate fine-grained asynchronous interactions between traces. Intuitively given a trace $t\in T$ and a value of the global clock $i\in \N$, a \tef $\tau$ outputs the value $\tau(i,t)$ of the local clock of trace $t$ at global time $i$.
If $T$ is a multiset of traces, a \emph{time evaluation function} for $T$ is a function $\tau\colon \N\times T \rightarrow \N$ that satisfies the following two conditions. We write $\tau(i)$ to denote the function $T\rightarrow\N$ defined by $t\mapsto \tau(i,t)$.
\begin{itemize}
\item \emph{stepwiseness} -- $\forall i\in \N \, \forall t\in T: \tau(i+1,t) \in \{ \tau(i,t), \tau(i,t)+1\}$,
\item \emph{strict monotonicity} -- $\forall i\in \N:\tau(i) \neq \tau(i+1)$.
\end{itemize}
A \tef is \emph{initial}, if $\tau(0,t)=0$ for each $t\in T$.

It was shown in \cite{GMOV22} that when \tefs are assumed to be synchronous, we obtain exactly synchronous $\teamltl$ as defined in \cite{kmvz18}.
In this section, we take a closer look on the connections between asynchronous $\teamltl$ and team-based logics with \tefs. We identify a logic with \tefs that corresponds almost exactly to asynchronous $\teamltl$ and to the left-flat fragment of asynchronous $\teamltl(\clor)$. This connection establishes the first non-trivial decidability result for logics with \tefs without putting heavy restrictions on \tefs.

We give the syntax $\teamctl$ with an additional synchronous next operator $\X$ that was shown in \cite{GMOV22} to be expressible in the logic using simple gadgets.
\[
\varphi \ddfn p \mid \neg p \mid (\phi\land\psi) \mid (\phi\lor\psi) \mid \X \phi \mid \X_\exists \phi \mid \X_\forall \phi \mid \G_\exists \phi \mid \G_\forall \phi \mid [\phi\U_\exists\psi] \mid [\phi\U_\forall\psi]
\]
We follow with the semantics. In \cite{GMOV22} $\teamctl^*$-formulae were evaluated with respect to pairs $(T,\tau)$, but since we conciser only $\teamctl$ here, we choose to internalise $\tau$ into $T$. The cases for propositional atoms, Boolean connectives, and $\X$ are as for asynchronous $\teamltl$ (see Definition \ref{def:strict}). Note that here the functions $\tau(i)$ take the role of the functions $f$ of Definition \ref{def:strict}.
	\begin{align*}
	&T \models\X_\exists \varphi  &&\Leftrightarrow && \text{there is an initial \tef $\tau$ s.t. $T[\tau(1),\infty] \models \varphi$}\\
	&T \models\X_\forall \varphi  &&\Leftrightarrow && \text{for all initial \tefs $\tau$, we have $T[\tau(1),\infty] \models \varphi$}\\
	&T \models\G_\exists \varphi  &&\Leftrightarrow && \text{there is an initial \tef $\tau$ s.t. $T[\tau(k),\infty] \models \varphi$, for all $k\in\N$}\\
	&T \models\G_\forall \varphi  &&\Leftrightarrow && \text{for all initial \tefs $\tau$, we have $T[\tau(k),\infty] \models \varphi$, for all $k\in\N$}\\
	&T\models[\phi\U_\exists\psi] &&\Leftrightarrow && \text{there is an initial \tef $\tau$ and $k\in\N$  s.t. } T[\tau(k),\infty]\models\psi \text{ and }\\
	&&&&&\quad\quad \forall m: 0 \leq m <k \Rightarrow T[\tau(m),\infty]\models\phi\\
		&T\models[\phi\U_\forall \psi] &&\Leftrightarrow && \text{for all initial \tefs $\tau$, $\exists k\in\N$  s.t. } T[\tau(k),\infty]\models\psi \text{ and }\\
	&&&&&\quad\quad \forall m: 0 \leq m <k \Rightarrow T[\tau(m),\infty]\models\phi
\end{align*}


We identify a collection of the above temporal operators that match as closely as possible with the operators of asynchronous $\teamltl$. For a collection of temporal operators $\mathcal{C}$, we write $\teamctl(\mathcal{C})$ to denote the logic built from propositional atoms by using $\land$, $\lor$, and the operators in $\mathcal{C}$.

In order to deal with the asynchronous until operator, we need to do two concessions. Firstly, we need to restrict ourselves to the left-flat fragment  (c.f. Definition \ref{def:leftflat}). 
 Secondly, instead of until, we use the strong release operator defined by $[\psi \M \varphi] \dfn [\varphi\U\varphi\land\psi]$. The  $\F$ and $\G$ modalities can be used without any restrictions.
Finally, we say that two formulas $\varphi$ and $\psi$ are \emph{fin-equivalent}, if $T\models \varphi \Leftrightarrow T\models \psi $ holds for all finite multiteams T.
With these restriction, we can prove the following theorem.
\begin{theorem}\label{thm:ctl}
For every left-flat-$\teamctl(\X, \G_\forall, \M_\exists, \clor)$-formula there exists an fin-equivalent 
$\teamltl^l(\clor)$ formula using $\M$ instead of $\U$, and vice versa.
\end{theorem}
\begin{proof}
It is easy to check that all the logics in question are downward closed (cf. the proof of Theorem \ref{thmdc} in the Appendix). We prove the equivalence with the left-flat-$\teamltl^l(\clor)$, which by Theorems \ref{lflat} and  \ref{disjnf} is equivalent with $\teamltl(\clor)$.

The translations simply swap $\M$ with $\M_\exists$ and $\G$ with $\G_\forall$.
Correctness of the translations can be proven by induction on the structure of formulae. The only non-trivial cases are the cases for strong release and globally.

The case for globally follow from the following chain of equivalences:
\begin{align*}
&T \models^l \G\varphi &&\Leftrightarrow&& \forall  f\colon T\rightarrow \mathbb{N}
		  \quad T[f,\infty] \models\varphi\\
		  	&  &&\Leftrightarrow&& \text{for all initial \tefs $\tau$, we have $T[\tau(k),\infty] \models \varphi$, for all $k\in\N$}\\
			&  &&\Leftrightarrow&& T \models\G_\forall \varphi.
\end{align*}
It is straightforward to check the second equivalence holds (recall that $\varphi$ is flat). The first and the last equivalence are simply the semantics of the respective operators. 
 
Assume $[\varphi\U\varphi\land\psi]$ is such that $\varphi$ is flat.
Now
\begin{align*}
&T \models^l [\varphi\U\varphi\land\psi] 
			&&\Leftrightarrow&&  \exists f\colon T\rightarrow \mathcal{P}(\mathbb{N})  \text{ such that } T[f,\infty] \models \varphi\land\psi\mbox{ and } \\
			&&&&& \forall f'\colon T'\rightarrow \mathcal{P}(\mathbb{N}) \text{ s.t. $f' < f$, we have } T'[f',\infty] \models\varphi\\
			&&&&& \text{where }T' \dfn \{t \in T \mid \max(f(t))\neq 0  \}\\
			&&&\Leftrightarrow&& \text{there is an initial \tef $\tau$ and $k\in\N$  s.t. } T[\tau(k),\infty]\models \varphi\land\psi \text{ and }\\
	&&&&&\quad\quad \forall m: 0 \leq m <k \Rightarrow T[\tau(m),\infty]\models\phi\\
	&&&\Leftrightarrow&& T \models [\varphi\U_\exists\varphi\land\psi].
\end{align*}
The first equivalence is the semantics of until and the second follows from flatness of $\varphi$, finiteness of $T$, and the induction hypothesis.
The last equivalence is the semantics of $\U_\exists$.
\end{proof}

\begin{corollary}
For every $\teamctl(\X, \G_\forall, \M_\exists)$-formula there exists a fin-equivalent $\teamltl$ formula using $\M$ instead of $\U$, and vice versa. (Note that there is no difference between $\teamltl$ and $\teamltl^l$.)
\end{corollary}

By combining Theorem \ref{thm:ctl} to Theorems \ref{lflat} and \ref{thm:complexity}, we obtain the following:
\begin{corollary}
The model checking problem of left-flat-$\teamctl(\X, \G_\forall, \M_\exists, \clor)$ restricted to finite teams is $\PSPACE$-complete.
\end{corollary}

We showed that over finite sets of traces the left-flat fragment of $\teamctl(\X, \G_\forall, \M_\exists, \clor)$ coincides with the left-flat fragment of $\teamltl(\clor)$ using $\M$ instead of $\U$. The following example shows that the same does not hold over arbitrary sets of traces.

\begin{example}
Let $T$ consist of the traces $\{p\}^k\{q\}^\omega$, $k\in\N$. Clearly $T\models^l p\U q$ in asynchronous $\teamltl$, but $T\not\models p\U_\exists q$ in $\teamctl$. 
\end{example}

\section{Conclusion}

We introduced a novel  set-based semantics for asynchronous \teamltl. We showed several results on the expressive power and complexity of the extensions of $\teamltl^{l}$ by the Boolean disjunction $\teamltl^{l}(\ovee)$ and by the Boolean negation $\teamltl^{l}(\cneg)$. In particular, our results show that the complexity properties  of the former logic are comparable to that of \LTL and that the left-downward closed fragment of the latter  has also decidable model-checking and satisfiability problems. See Table \ref{table:exp} on page \pageref{table:exp} for an overview of our expressivity results and Table \ref{table:comp} for our complexity results. We obtained these results on  $\teamltl^{l}(\ovee)$ and $\teamltl^{l}(\cneg)$ via normal forms that also allowed us to relate the  expressive power of these logics to the corresponding logics in the strict semantics.
Our results show that, while the synchronous TeamLTL can be viewed as a fragment of second-order logic, the asynchronous $\teamltl(\clor)$ under the lax semantics is a sublogic of HyperLTL (see \cite{DBLP:conf/post/ClarksonFKMRS14} for a definition). Furthermore, our decidability results show, e.g, that it will probably be possible to devise a complete proof system for the logic.
We conclude with open questions:
\begin{itemize}
\item Does Theorem \ref{qfnf} extend to all formulae of  $\teamltl^{l}(\cneg)$? Note that any quasi-flat--$\teamltl(\sim)$-formula can be rewritten in \hyltl .
\item Can the result  (iii)  of Theorem \ref{thm:complexity} be accompanied by an matching lower bound (i.e.,  $\mathrm{TOWER(poly)}$-hardness result)?
\item Can a syntactic characterisation (similar to Corollary \ref{cor:flatness}) be obtained for the downward closed fragment of $\teamltl^{l}(\cneg)$?
We believe that $\teamltl^{l}(\clor)$ is a promising candidate, as its extensions with infinite conjunctions and disjunctions suffices for all downward closed properties of teams.
\item What is the complexity of model checking for $\teamltl(\clor)$ under the strict semantics?
\end{itemize}

\bibliographystyle{plain}
\bibliography{biblio}

\appendix

\section{Appendix}\label{a:one}
\setbased*
\begin{proof}  
  The proof is by induction on the structure of $\varphi$.
  
  Suppose $\varphi$ is a literal.
  By definition, $T \models^l \varphi$ if only if $t \models \varphi$ for all $t \in T$. Now, seeing as $\varphi$ is a literal of $\ap$, the latter is equivalent with $t \models \varphi$ for all $t \in T_{\upharpoonright\ap}$. This, by definition, is equivalent with $T_{\upharpoonright\ap} \models^l \varphi$.

  Suppose $\varphi = \psi_1 \wedge \psi_2$. By definition, $T\models^l \varphi$ if and only if $T\models^l \psi_1$ and $T\models^l \psi_2$. By the induction hypothesis, the latter holds if and only if $T_{\upharpoonright\ap} \models^l \psi_1$ and $T_{\upharpoonright\ap} \models^l \psi_2$, which by definition is equivalent with $T_{\upharpoonright\ap} \models^l \varphi$.
  
  Suppose $\varphi = \psi' \vee \psi''$. By definition, $T \models^l \varphi$ if and only if there are subteams $T' \cup T'' = T$ such that $T' \models^l \psi'$ and $T'' \models^l \psi''$. By the induction hypothesis, the latter two claims are equivalent with $T'_{\upharpoonright\ap} \models^l \psi'$ and $T''_{\upharpoonright\ap} \models^l \psi''$. Now $T'_{\upharpoonright\ap} \cup T''_{\upharpoonright\ap} = T_{\upharpoonright\ap}$, hence $T_{\upharpoonright\ap} \models^l \varphi$. For the converse, assume that $T_{\upharpoonright\ap} \models^l \varphi$ is witnessed by subteams $T'_{\upharpoonright\ap}$ and $T''_{\upharpoonright\ap}$. Now it is easy to check that $T'=\{t\in T \mid t_{\upharpoonright\ap} \in T'_{\upharpoonright\ap}\}$
and   $T''=\{t\in T \mid t_{\upharpoonright\ap} \in T''_{\upharpoonright\ap}\}$ witness $T \models^l \varphi$.

  Suppose $\varphi = \X\psi$. By definition, $T \models^l \varphi$ if and only if $T[1,\infty] \models^l \psi$, which by the induction hypothesis is equivalent with $T[1,\infty]_{\upharpoonright\ap} \models^l \psi$. Note that $T[1,\infty]_{\upharpoonright\ap} = T_{\upharpoonright\ap}[1,\infty]$, whereby $T \models^l \X\psi$ holds if and only if $T_{\upharpoonright\ap} \models^l \X\psi$ holds.
  
  Suppose $\varphi = \G\psi$. By definition $T \models^l \varphi$ is equivalent with that  $T[f,\infty] \models^l \psi$ for all functions $f\colon T \to \mathcal{P}(\N^+)$. By the induction hypothesis the latter is equivalent with $T[f,\infty]_{\upharpoonright\ap} \models^l \psi$ for functions $f$ as before. Now for any $f'\colon T_{\upharpoonright\ap} \to \mathcal{P}(\N^+)$ there is some $f$, such that $T[f,\infty]_{\upharpoonright\ap} = T_{\upharpoonright\ap}[f',\infty]$, since we can pick the function $f(s) = f'(t)$ for all $s\in T$ such that $s_{\upharpoonright\ap} = t$. Similarly, for each $f$ we obtain a corresponding $f'$ by taking its restriction to $\ap$. Hence $T[f,\infty] \models^l \psi$ holds for all $f$ if and only if $T_{\upharpoonright\ap}[f',\infty] \models^l \psi$ holds for all $f'$, and therefore $T \models^l \G\psi$ is equivalent with $T_{\upharpoonright\ap} \models^l \G\psi$.
  
  Suppose $\varphi = \psi_1 \U \psi_2$. Assume $T \models^l \varphi$. By definition there is a function $f_1\colon T \to \mathcal{P}(\N^+)$ such that $T[f_1,\infty] \models^l \psi_1$ and for all $f_2 < f_1$ it holds that $T^0[f_2,\infty] \models^l \psi_2$, where $T^0 \dfn \{t \in T\mid \max(f(t)) \neq 0\}$. By the induction hypothesis then $T[f_1,\infty]_{\upharpoonright\ap} \models^l \psi_1$ and $T^0[f_2,\infty]_{\upharpoonright\ap} \models^l \psi_2$ for $f_1$ and $f_2$ as previously. We define the function $f_1'\colon T_{\upharpoonright\ap} \to \mathcal{P}(\N^+)$ by setting $f_1'(s) \dfn \bigcup_i f_1(t^i)$, where $t^i\in T$ are such that $t^i_{\upharpoonright\ap} = s$. Now $T_{\upharpoonright\ap}[f_1',\infty] = T[f_1,\infty]_{\upharpoonright\ap}$. Furthermore, by a similarly defined $f_2'$, we get that $T^0_{\upharpoonright\ap}[f_2',\infty] = T^0[f_2,\infty]_{\upharpoonright\ap}$. Thus $T[f_1,\infty]_{\upharpoonright\ap} \models^l \psi_1$ if and only if $T_{\upharpoonright\ap}[f_1',\infty] \models^l \psi_1$ and for all $f_2<f_1$ it holds that $T^0[f_2,\infty]_{\upharpoonright\ap} \models^l \psi_2$ if and only if $T^0_{\upharpoonright\ap}[f_2,\infty] \models^l \psi_2$. Therefore $T_{\upharpoonright\ap} \models^l \varphi$. The converse follows analogously.
  
  Suppose $\varphi = \sim \psi$. By definition $T \models^l \varphi$ if and only if $T \not\models^l \psi$, which in turn is equivalent with $T_{\upharpoonright\ap} \not\models^l \psi$ by the induction hypothesis. This, again, is equivalent with $T_{\upharpoonright\ap} \models^l \varphi$, due to the definition.
\end{proof}

\basicprop*
 \begin{proof}
 The proofs proceed by induction over the structure of the formulae. Note that while downward closure follows from flatness, we need that the induction steps work with the weaker assumption of downward closure for  the result to generalise to non-flat extensions of the logic.
 
\textbf{Downward closure:} Let $\varphi\in\teamltl$ be a formula and $T,S$ teams such that $S \subseteq T$ and $T\lmodels \varphi$. We need to show that $S\lmodels \varphi$ as well.

For atomic $\varphi$ the claim is immediately true: $T\models^l \varphi$ if and only if $t \models \varphi$ for all $t \in T$, which also holds for all $t \in S$, and thus $S \models^l \varphi$.

For conjunction the claim follows immediately from the induction hypothesis. Let's consider the case of disjunction. Suppose $T \models^{l} \varphi \vee \psi$. By definition then there are $T_1, T_2 \subseteq T$ s.t. $T_1 \models^{l} \varphi$ and $T_2 \models^{l} \psi$. By the induction hypothesis $S\cap T_1 \models^{l} \varphi$ and $S\cap T_2 \models^{l} \psi$, since $S\cap T_1 \subseteq T_1$ and $S\cap T_2\subseteq T_2$. Furthermore $(S\cap T_1)\cup (S\cap T_2) = S$, and therefore $S \models^l \varphi\vee \psi$.
     
The case for $\X$ is straightforward. Suppose $T \models^l \X \varphi$. By definition $T[1,\infty] \models^l \varphi$. Since $S[1,\infty]\subseteq T[1,\infty]$ follows from $S\subseteq T$, we obtain $S[1,\infty] \models^l \varphi$ by the induction hypothesis. Thus $S\models \X \varphi$.

Next we suppose $T \models^l \G\varphi$. By definition then for all $f\colon T\rightarrow \mathcal{P}(\mathbb{N})^+$ it holds that $T[f,\infty] \models^l \varphi$. Now,  for all  $f\colon S\rightarrow \mathcal{P}(\mathbb{N})^+$, $S[f,\infty] \subseteq T[f',\infty]$, where $f'$ is any extension of $f$ to $T$. Hence by the induction hypothesis $S[f,\infty] \models^l \varphi$ for all  $f$ and $S \models^l \G\varphi$.

Suppose then that $T \models^l \varphi\U\psi$. For any function $h\colon T\rightarrow \mathcal{P}(\mathbb{N})^+$, let $h_S$ denote the reduct of $h$ to the domain $S$. From $S\subseteq T$, we get 
\begin{equation}\label{eq:U}
S[h_S,\infty]\subseteq T[h,\infty].
\end{equation}
By definition, there is a function $f\colon T\rightarrow \mathcal{P}(\mathbb{N})^+$ such that $T[f,\infty] \models^l \psi$. Moreover, for all $f'\colon T_0\rightarrow \mathcal{P}(\mathbb{N})^+$ such that $f' < f$, we have $T_0[f',\infty] \models^l \varphi$, where $T_0 \dfn \{t \in T\mid \max(f(t)) \neq 0\}$. By the induction hypothesis and \eqref{eq:U}, we have that $S[f_S,\infty] \models^l \psi$. Moreover, for all $g\colon S_0\rightarrow \mathcal{P}(\mathbb{N})^+$, where $S_0 \dfn \{t \in S\mid \max(f_S(t)) \neq 0\}$, such that $g < f_S$,
$S[g,\infty] \models^l \varphi$ follows by the induction hypothesis and the fact that every $g$ is equal to $f'_S$ for a suitable $f'$. Thus $S \models^l \varphi\U\psi$.

\textbf{Empty team property:} Suppose $T = \emptyset$. The claim is clear for atomic formulae; since the team is empty, $p \in t(0)$ and $p \notin t(0)$ holds for all $p \in \mathrm{AP}$ and $t \in T$. The cases for conjunction and disjunction follow immediately from the induction hypothesis. The cases for temporal operators follow immediately from the induction hypotheses as well, since $\emptyset[f,\infty] = \emptyset$ for any $f\colon T\rightarrow \mathcal{P}(\mathbb{N})^+$.

\textbf{Flatness:} We prove by induction on $\varphi$ that $T\models\phi$ if $\{t\} \models\phi$ for all $t\in T$, for every team $T$. The only if direction follows directly from downward closure.

The case for atomic formulae holds by definition, and the cases for conjunction and the next step operator follows directly from the induction hypothesis. Let us consider the case for disjunction. Assume that for all $t \in T$, $\{t\} \models \varphi$ or $\{t\}\models \psi$. Let $T_1$ and $T_2$ be the sets of traces in $T$ that satisfy $\varphi$ and $\psi$, respectively. Clearly $T_1\cup T_2=T$, and by induction hypothesis $T_1 \models^{l} \varphi$ and $T_2 \models^{l} \psi$. Thus $T \models^{l} \varphi \vee \psi$.

For the case of until, suppose $\{t\} \models^l \varphi\U\psi$ for all $t \in T$. Now, for each $t\in T$, there exists a function $f_t\colon \{t\} \to \mathcal{P}(\N)^+$ such that $\{t\}[f_t,\infty] \models \psi$ and for all intermediary $f_t'\colon \{t\} \to \mathcal{P}(\N)^+$, defined for such traces $t$ where $\max(f(t)) \neq 0$, such that $f'_t < f_t$ it holds that $\{t\}[f'_t,\infty] \models \varphi$. We define the functions $f\colon T \to \mathcal{P}(\N)^+$, and $g\colon T' \to \mathcal{P}(\N)^+$ through $f(t) \dfn f_t(t)$ and $g(t) \dfn \{j \in \N \mid j < \sup(f(t)) \}$. By the induction hypothesis $T[f,\infty] \models^{l} \psi$ and $T'[g,\infty] \models^{l} \varphi$. Let $f'\colon T' \to \mathcal{P}(\N)^+$ be any function such that $f' < f$. If we can show that $T'[f',\infty] \models^{l} \psi$, we obtain $T \models \varphi\U\psi$ and are done. To that end, clearly $T'[f',\infty] \subseteq T'[g,\infty]$, and thus we obtain $T'[f',\infty] \models^{l} \psi$ from downward closure.

Finally for the case for $\G$,  suppose $\{t\} \models^l \G\varphi$ for all $t \in T$. Now for each trace $t$ and function $f\colon \{t\}\rightarrow \mathcal{P}(\mathbb{N})^+$ it holds that $\{t\}[f,\infty] \models^l \varphi$. Now by the induction hypothesis, for all functions $F\colon T\rightarrow \mathcal{P}(\mathbb{N})^+$ such that $F(t) \dfn f(t)$ it holds that $T[F,\infty] \models^l \varphi$. Now the functions $F$ are all possible functions. Thus $T \models^l \G\varphi$.

\textbf{Singleton equivalence:} We prove by induction on $\varphi$ that, for every trace $t$, $\{t\} \models^l \varphi$ if and only if $t \models^l \varphi$.

The case for literals is stated in the definition, whereas the case for conjunction follows directly from the induction hypothesis. Hence, consider the next step operator. Suppose $\varphi = \X\psi$. Now by definition $\{t\} \models^l \varphi$ is equivalent with  $\{t\}[1,\infty] \models^l \psi$, which in turn holds if and only if $t[1,\infty] \models \psi$ by induction hypothesis. By the definition of the next step operator the latter is equivalent with $t \models \X\psi$.

Suppose $\varphi = \psi \vee \theta$. By definition $\{t\} \models^l \varphi$ is equivalent with $\{t\} \models^l \psi$ or $\{t\} \models^l \theta$. By the induction hypothesis the latter two are equivalent with $t \models \psi$ or $t \models \theta$, and thus by definition $t \models \psi \vee \theta$ if and only if $\{t\} \models^l \varphi$.

Suppose $\varphi = \psi_1\U\psi_2$. Assume $\{t\} \models^l \varphi$. By definition there is a function $f\colon \{t\} \to \mathcal{P}(\N)^+$ such that $\{t\}[f,\infty]\models^l \psi_2$ and for all intermediary functions $f'\colon T' \to \mathcal{P}(\N)^+$ it holds that $T'[f',\infty] \models^l \psi_1$, where $T' \dfn \{ t \}$, if $f(t) \neq \{ 0 \}$ and otherwise $T' = \emptyset$. We assume $f(t) \neq \{ 0 \}$, as the other case is trivial.
Let $k \dfn \min(f(t))$. By induction hypothesis and downward closure $t[k,\infty] \models \psi_2$. Now for every singleton-valued function $f'\colon T' \to \mathcal{P}(\N)^+$ defined by $f'(t) \dfn \{k'\}$, such that $k' < k$, it holds that $\{t\}[f',\infty] \models^l \psi_1$. Hence by the induction hypothesis, for all $k'  < k$ it holds that $t[k',\infty] \models \psi_1$. Thus $t\models \varphi$.

Now assume $t \models \varphi$. By definition there exists a number $k \geq 0$ such that $t[k,\infty] \models \psi_2$ and for all $0 \leq k' < k$ it holds that $t[k',\infty] \models \psi_1$. Thus we can define a function $f\colon \{t\} \to \mathcal{P}(\N)^+$ such that $f(t) \dfn \{k\}$ and functions $f'\colon \{t\} \to \mathcal{P}(\N)^+$ such that $f'(t) \dfn \{ k' \}$. Now by the induction hypothesis $\{t\}[f,\infty] \models^l \psi_2$ and $\{t\}[f',\infty] \models^l \psi_1$, and furthermore by flatness the latter actually holds for all intermediary functions $g$. Therefore $\{t\} \models^l \varphi$.

Suppose $\varphi = \G \psi$. Assume $\{t\} \models^l \varphi$. By definition for all functions $f\colon \{t\} \to \mathcal{P}(\N)^+$ it holds that $\{t\}[f,\infty] \models^l \psi$. Especially this holds for every function $f_k$ such that $f_k(t) \dfn \{k\}$. By the induction hypothesis then $t[k,\infty] \models \psi$ for all $k$. Thus by definition $t \models \G\psi$.

Now assume $t \models \varphi$. By definition for all $k\geq 0$ it holds that $t[k,\infty] \models \psi$. By the induction hypothesis it follows that $\{t\}[f,\infty] \models^l \psi$ for all $f\colon \{t\} \to \mathcal{P}(\N)^+$ such that $f = \{ k \}$. Thus, by flatness $\{t\}[f',\infty] \models^l \psi$ for all functions $f'\colon \{t\} \to \mathcal{P}(\N)^+$. Thus $\{t\} \models^l \G\psi$.      
 \end{proof}
 
\leftflat*
 \begin{proof}
In order to show \eqref{i} it suffices to extend the proof of Theorem \ref{thmdc} with a  case for $\clor$: Let $T$ be a team of traces and let $S \subseteq T$. Suppose $T \models^l \psi \clor \varphi$. By definition $T \models^l \psi$ or $T \models^l \varphi$. Without loss of generality we may assume $T \models^l \psi$, which entails by the induction hypothesis that $S \models^l \psi$ and  $S \models^l \psi \clor \varphi$.

The proof of  claim \eqref{ii} is a simple induction on the structure of $\phi$.
We show the claim for  $\phi= \alpha\U\psi$, where $\alpha $ is an \LTL-formula. Assume  $T\models \phi$. Then there exists $f\colon T \to \N$ such that   $T[f,\infty] \models \psi$ and, by flatness, for all $(i,t)\in T$ and $k<f((i,t))$ it holds that $\{(i,t[k,\infty])\} \models \alpha$. Define $F\colon T \to \mathcal{P}(\N)^+$ by $F(t) \dfn \{f((i,t)) \mid i\in\N\}$. It is easy to check that $\support(T[f,\infty]) = \support(T)[F,\infty]$. Now by application of the induction hypothesis $\support(T)[F,\infty]\models^l \psi$. Likewise, for all $F'<F$, where $F'\colon T \to \mathcal{P}(\N)^+$ is defined similarly to $F$, $\support(T)[F',\infty] \models^l \alpha$ follows by the flatness of $\alpha$. The proof of the converse implication is similar. Assume  $\support(T)\models^l \phi$ and let $G\colon \support(T) \to \mathcal{P}(\N)^+$ be such that $\support(T)[G,\infty] \models^l \psi$. By downward closure we may assume that $G$ is single valued. Now it is easy to pick $g\colon T \to \N$ such that $\support(T[g,\infty]) = \support(T)[G,\infty]$. From the induction hypothesis it follows that $T[g,\infty] \models \psi$. Just like above, using the fact that $\alpha$ is flat, it follows that  $T\models \phi$.
 \end{proof}
 
\corflatness*
\begin{proof}
Let $\varphi \in \teamltl^l(\clor)$ be flat, and let $\Clor_i \psi_i$ be an equivalent formula given by Theorem \ref{disjnf}, where $\psi_i$ are $\teamltl^l$-formulae. 
The following equivalences hold:
\[
T \lmodels \varphi \quad\Leftrightarrow\quad \forall t\in T : \{t\} \lmodels \Clor_i \psi_i  \quad\Leftrightarrow\quad \forall t\in T : \{t\} \lmodels \bigvee_i \psi_i \quad\Leftrightarrow\quad T \lmodels \bigvee_i \psi_i
\]
The first equivalence follows from the above mentioned Theorem, and flatness of $\varphi$.
The second equivalence follows, for it is easy to check that, for logics that have the empty team property, $\Clor_i$ and $\bigvee_i$ are interchangeable over singleton teams. The last equivalence follows from the flatness of $\teamltl^l$ (Theorem \ref{thmdc}).
\end{proof}

\corlaxstrict*
\begin{proof}
 Let $\varphi$ is a $\teamltl^l(\clor)$-formula. By Theorem \ref{disjnf}, $\varphi$ can be equivalently written as a disjunction $\Clor_i\alpha_i$ of $\teamltl^l$-formulae. Now, for each multiteam $T$,
\[
\support(T)\lmodels\varphi \quad\Leftrightarrow\quad \support(T)\lmodels \Clor_i\alpha_i \quad\Leftrightarrow\quad T\models \Clor_i\alpha_i,
\]
where the last equivalence is due to Theorem \ref{lflat} and the fact that the formulae $\alpha_i$ are flat. Hence, for any given $\teamltl^l(\clor)$-formula, the normal form formula $\Clor_i\alpha_i$ is, in fact, an equivalent $\teamltl(\clor)$ formula, from which $\teamltl^{l}(\clor) \leq \teamltl(\clor)$ follows.

For showing the strict inclusion, we generalize the result from Example \ref{ex1} and show that for $\teamltl(\clor)$ formula $\G(p\ovee q)$ there exists no equivalent $\teamltl^{l}(\clor)$ formula.
For a contradiction, suppose that $\varphi\in \teamltl^l(\clor)$ is equivalent with $\G(p\ovee q)$. By Theorem \ref{disjnf}, we may assume that $\varphi$ is a disjunction $\Clor_i\alpha_i$ of $n$ $\teamltl^l$-formulae. Define $t_i \dfn \{p\}^i \{q\}^{\omega}$ for $i\leq n+1$.
Now clearly $\{(1,t_i)\} \models \G(p\ovee q)$ and thus $\{t_i\} \lmodels \varphi$, for each $i$. By the semantics of $\clor$ this implies that for each $i$ there exists $j_i\leq n$ such that $\{t_i\} \lmodels \alpha_{j_i}$. Now from the pigeonhole principle, there exists $1\leq k < l \leq n+1$ such that $j_k=j_l$. Thus $\{t_k\} \lmodels \alpha_{j_k}$ and $\{t_l\} \lmodels \alpha_{j_k}$, from which $\{t_k,t_l\} \lmodels \alpha_{j_k}$ follows, by flatness of $\alpha_{j_k}$. Thus $\{t_k,t_l\} \lmodels \varphi$ and $\{(1,t_k),(1,t_l)\} \models \G(p\ovee q)$, which is clearly false.
\end{proof}

\end{document}

%% file: macros.tex
\newcommand{\LTL}{\protect\ensuremath{\mathrm{LTL}}\xspace}
\newcommand{\CTL}{\protect\ensuremath{\mathrm{CTL}}\xspace}
\newcommand{\QPTL}{\protect\ensuremath{\mathrm{QPTL}}\xspace}
\newcommand{\ML}{\protect\ensuremath{\mathrm{ML}}\xspace}
\newcommand{\teamltl}{\protect\ensuremath{\mathrm{TeamLTL}}\xspace}
\newcommand{\lteamltl}{\protect\ensuremath{\mathrm{TeamLTL}^l}\xspace}

\newcommand{\hyctl}{\protect\ensuremath{\mathrm{HyperCTL^*}}\xspace}
\newcommand{\hyltl}{\protect\ensuremath{\mathrm{HyperLTL}}\xspace}

\newcommand{\HQPTLP}{\protect\ensuremath{\mathrm{HyperQPTL\textsuperscript{\hskip-2pt\small +}}}\xspace}
\newcommand{\HQPTL}{\protect\ensuremath{\mathrm{HyperQPTL}}\xspace}

\newcommand{\dfn}{\mathrel{\mathop:}=}
\newcommand{\ddfn}{\mathrel{\mathop{{\mathop:}{\mathop:}}}=}
\newcommand{\lmodels}{\mathrel{\models^l}}

\newcommand{\ldot}{\mathpunct{.}}
\newcommand{\pow}[1]{2^{#1}}

\newcommand{\LL}{\mathcal{L}}

\newcommand{\support}{\mathrm{support}}

\newcommand{\U}{\LTLuntil}
\newcommand{\W}{\LTLweakuntil}
\newcommand{\X}{\LTLnext}
\newcommand{\G}{\LTLg}
\newcommand{\F}{\LTLf}
\newcommand{\R}{\LTLrelease}

\DeclareMathOperator{\M}{\mathcal{M}}

\newcommand{\pathvars}{\mathcal{V}}
\newcommand{\pathassign}{\Pi}

\newcommand{\N}{\mathbb N}
\newcommand{\ap}{\mathrm{AP}}

\newcommand{\traces}{\mathrm{Traces}}

\newcommand{\truth}[1]{\llbracket #1 \rrbracket}

\newcommand{\NP}{\protect\ensuremath{\mathrm{NP}}} 
\newcommand{\PSPACE}{\protect\ensuremath{\mathrm{PSPACE}}}

\newcommand{\EXPSPACE}{\protect\ensuremath{\mathrm{EXPSPACE}}}
\newcommand{\NEXPTIME}{\protect\ensuremath{\mathrm{NEXPTIME}}}
\newcommand{\PL}{\protect\ensuremath{\mathrm{PL}}}
\newcommand{\Ptime}{\protect\ensuremath{\mathrm{P}}}

\renewcommand{\phi}{\varphi}
\renewcommand{\epsilon}{\varepsilon}

\newcommand{\kK}{\mathrm{K}}

\newcommand{\clor}{\varovee}
\newcommand{\lando}{\rotatebox[origin=c]{180}{$\,\varovee\,$}}

\DeclareMathOperator*{\Clor}{\scalerel*{\ovee}{\sum}}

\newcommand{\cneg}{\sim}

\newcommand{\teamctl}{\protect\ensuremath{\mathrm{Team}\CTL}}

\newcommand{\tef}{tef\xspace}
\newcommand{\tefs}{tefs\xspace}

%% file: Aspects_of_Asynchronicity.bbl
\begin{thebibliography}{10}

\bibitem{BaumeisterCBFS21}
Jan Baumeister, Norine Coenen, Borzoo Bonakdarpour, Bernd Finkbeiner, and
  C{\'{e}}sar S{\'{a}}nchez.
\newblock A temporal logic for asynchronous hyperproperties.
\newblock In {\em {CAV} {(1)}}, volume 12759 of {\em Lecture Notes in Computer
  Science}, pages 694--717. Springer, 2021.

\bibitem{DBLP:conf/post/ClarksonFKMRS14}
Michael~R. Clarkson, Bernd Finkbeiner, Masoud Koleini, Kristopher~K. Micinski,
  Markus~N. Rabe, and C{\'{e}}sar S{\'{a}}nchez.
\newblock Temporal logics for hyperproperties.
\newblock In {\em {POST} 2014}, pages 265--284, 2014.

\bibitem{clarkson}
Michael~R. Clarkson and Fred~B. Schneider.
\newblock Hyperproperties.
\newblock {\em Journal of Computer Security}, 18(6):1157--1210, 2010.

\bibitem{DBLP:conf/lics/CoenenFHH19}
Norine Coenen, Bernd Finkbeiner, Christopher Hahn, and Jana Hofmann.
\newblock The hierarchy of hyperlogics.
\newblock In {\em LICS 2019}, pages 1--13. {IEEE}, 2019.

\bibitem{DBLP:journals/acta/FinkbeinerHLST20}
Bernd Finkbeiner, Christopher Hahn, Philip Lukert, Marvin Stenger, and Leander
  Tentrup.
\newblock Synthesis from hyperproperties.
\newblock {\em Acta Informatica}, 57(1-2):137--163, 2020.

\bibitem{galliani12}
Pietro Galliani.
\newblock Inclusion and exclusion dependencies in team semantics: On some
  logics of imperfect information.
\newblock {\em Annals of Pure and Applied Logic}, 163(1):68 -- 84, 2012.

\bibitem{GMOV22}
Jens~Oliver Gutsfeld, Arne Meier, Christoph Ohrem, and Jonni Virtema.
\newblock Temporal team semantics revisited.
\newblock In Christel Baier and Dana Fisman, editors, {\em {LICS} '22: 37th
  Annual {ACM/IEEE} Symposium on Logic in Computer Science, Haifa, Israel,
  August 2 - 5, 2022}, pages 44:1--44:13. {ACM}, 2022.

\bibitem{GutsfeldMO21}
Jens~Oliver Gutsfeld, Markus M{\"{u}}ller{-}Olm, and Christoph Ohrem.
\newblock Automata and fixpoints for asynchronous hyperproperties.
\newblock {\em Proc. {ACM} Program. Lang.}, 5({POPL}):1--29, 2021.

\bibitem{HKVV18}
Miika Hannula, Juha Kontinen, Jonni Virtema, and Heribert Vollmer.
\newblock Complexity of propositional logics in team semantic.
\newblock {\em {ACM} Trans. Comput. Log.}, 19(1):2:1--2:14, 2018.

\bibitem{HellaKMV19}
Lauri Hella, Antti Kuusisto, Arne Meier, and Jonni Virtema.
\newblock Model checking and validity in propositional and modal inclusion
  logics.
\newblock {\em J. Log. Comput.}, 29(5):605--630, 2019.

\bibitem{KontinenMSV14}
Juha Kontinen, Julian{-}Steffen M{\"{u}}ller, Henning Schnoor, and Heribert
  Vollmer.
\newblock Modal independence logic.
\newblock In Rajeev Gor{\'{e}}, Barteld~P. Kooi, and Agi Kurucz, editors, {\em
  Advances in Modal Logic 10, invited and contributed papers from the tenth
  conference on "Advances in Modal Logic," held in Groningen, The Netherlands,
  August 5-8, 2014}, pages 353--372. College Publications, 2014.

\bibitem{KS21}
Juha Kontinen and Max Sandstr{\"{o}}m.
\newblock On the expressive power of teamltl and first-order team logic over
  hyperproperties.
\newblock In {\em WoLLIC}, volume 13038 of {\em Lecture Notes in Computer
  Science}, pages 302--318. Springer, 2021.

\bibitem{KrebsMV15}
Andreas Krebs, Arne Meier, and Jonni Virtema.
\newblock A team based variant of {CTL}.
\newblock In Fabio Grandi, Martin Lange, and Alessio Lomuscio, editors, {\em
  22nd International Symposium on Temporal Representation and Reasoning, {TIME}
  2015, Kassel, Germany, September 23-25, 2015}, pages 140--149. {IEEE}
  Computer Society, 2015.

\bibitem{kmvz18}
Andreas Krebs, Arne Meier, Jonni Virtema, and Martin Zimmermann.
\newblock {Team Semantics for the Specification and Verification of
  Hyperproperties}.
\newblock In Igor Potapov, Paul Spirakis, and James Worrell, editors, {\em MFCS
  2018}, volume 117, pages 10:1--10:16, Dagstuhl, Germany, 2018. Schloss
  Dagstuhl--Leibniz-Zentrum fuer Informatik.

\bibitem{Luck18b}
Martin L{\"{u}}ck.
\newblock Axiomatizations of team logics.
\newblock {\em Ann. Pure Appl. Log.}, 169(9):928--969, 2018.

\bibitem{Luck18a}
Martin L{\"{u}}ck.
\newblock On the complexity of team logic and its two-variable fragment.
\newblock In {\em {MFCS}}, volume 117 of {\em LIPIcs}, pages 27:1--27:22.
  Schloss Dagstuhl - Leibniz-Zentrum f{\"{u}}r Informatik, 2018.

\bibitem{LUCK2020}
Martin L{\"{u}}ck.
\newblock On the complexity of linear temporal logic with team semantics.
\newblock {\em Theoretical Computer Science}, 2020.

\bibitem{LuckPHD20}
Martin L{\"{u}}ck.
\newblock {\em Team logic: axioms, expressiveness, complexity}.
\newblock PhD thesis, University of Hanover, Hannover, Germany, 2020.

\bibitem{DBLP:journals/jcs/McLean92}
John McLean.
\newblock Proving noninterference and functional correctness using traces.
\newblock {\em Journal of Computer Security}, 1(1):37--58, 1992.

\bibitem{pnueli}
Amir Pnueli.
\newblock The temporal logic of programs.
\newblock In {\em 18th Annual Symposium on Foundations of Computer Science},
  pages 46--57. {IEEE} Computer Society, 1977.

\bibitem{MarkusThesis}
Markus~N. Rabe.
\newblock {\em A Temporal Logic Approach to Information-Flow Control}.
\newblock PhD thesis, Saarland University, 2016.

\bibitem{DBLP:conf/sp/Roscoe95}
A.~W. Roscoe.
\newblock {CSP} and determinism in security modelling.
\newblock In {\em Proceedings of the 1995 {IEEE} Symposium on Security and
  Privacy, Oakland, California, USA, May 8-10, 1995}, pages 114--127. {IEEE}
  Computer Society, 1995.

\bibitem{SC85}
A.~P. Sistla and E.~M. Clarke.
\newblock The complexity of propositional linear temporal logics.
\newblock {\em J. ACM}, 32(3):733--749, jul 1985.

\bibitem{vaananen07}
Jouko V\"a\"an\"anen.
\newblock {\em Dependence Logic}.
\newblock Cambridge University Press, 2007.

\bibitem{Virtema17}
Jonni Virtema.
\newblock Complexity of validity for propositional dependence logics.
\newblock {\em Inf. Comput.}, 253:224--236, 2017.

\bibitem{DBLP:journals/corr/abs-2010-03311}
Jonni Virtema, Jana Hofmann, Bernd Finkbeiner, Juha Kontinen, and Fan Yang.
\newblock Linear-time temporal logic with team semantics: Expressivity and
  complexity.
\newblock {\em CoRR}, abs/2010.03311, 2020.

\bibitem{VBHKF20}
Jonni Virtema, Jana Hofmann, Bernd Finkbeiner, Juha Kontinen, and Fan Yang.
\newblock Linear-time temporal logic with team semantics: Expressivity and
  complexity.
\newblock In Mikolaj Bojanczyk and Chandra Chekuri, editors, {\em 41st {IARCS}
  Annual Conference on Foundations of Software Technology and Theoretical
  Computer Science, {FSTTCS} 2021, December 15-17, 2021, Virtual Conference},
  volume 213 of {\em LIPIcs}, pages 52:1--52:17. Schloss Dagstuhl -
  Leibniz-Zentrum f{\"{u}}r Informatik, 2021.

\bibitem{DBLP:conf/csfw/ZdancewicM03}
Steve Zdancewic and Andrew~C. Myers.
\newblock Observational determinism for concurrent program security.
\newblock In {\em 16th {IEEE} Computer Security Foundations Workshop {(CSFW-16}
  2003), 30 June - 2 July 2003, Pacific Grove, CA, {USA}}, page~29. {IEEE}
  Computer Society, 2003.

\end{thebibliography}
